\documentclass[12pt, onecolumn, draftclsnofoot]{IEEEtran}
\usepackage{cite}
\usepackage{graphicx} 
\usepackage{epsf}
\usepackage{subfigure}
\usepackage{amsmath}
\usepackage{comment}
\usepackage{amssymb}
\usepackage{array}
\usepackage{setspace}
\usepackage{amsthm,amssymb}
\usepackage[ruled,lined]{algorithm2e}
\usepackage{multirow} 
\usepackage{tabularx}
\usepackage{longtable}
\usepackage{xfrac}
\usepackage{breqn}
\usepackage{bbm}
\usepackage{enumerate}
\usepackage{algorithmic}
\usepackage{dsfont}
\usepackage{float}
\usepackage{diagbox}
\usepackage{mathtools}
\usepackage{adjustbox}
\usepackage{color}

\graphicspath{{Fig/},{fig/}}

\DeclareMathOperator{\Emean}{\mathbb{E}}

\DeclareMathOperator{\X}{\rm{x}}

\DeclareMathOperator{\s}{\rm{s}}
\DeclareMathOperator{\M}{\rm{m}}
\DeclareMathOperator{\CM}{\rm{CM}}
\DeclareMathOperator{\DM}{\rm{DM}}
\DeclareMathOperator{\e}{\rm{e}}

\DeclareMathOperator{\C}{\rm{CU}}
\DeclareMathOperator{\D}{\rm{DU}}

\DeclareMathOperator{\SINR}{\Gamma}

\DeclareMathOperator{\pr}{\mathbbm{P}}
\DeclareMathOperator{\U}{\rm{u}}
\DeclareMathOperator{\ee}{\rm{E2E}}
\DeclareMathOperator{\ow}{\rm{ow}}
\DeclareMathOperator{\tti}{\rm{TTI}}

\DeclareMathOperator{\co}{\rm{CO}}
\DeclareMathOperator{\re}{\rm{re}}

\newtheorem{prop}{Proposition}
\newtheorem{theorem}{Theorem}

\newtheorem{corollary}{Corollary}

\begin{document}

\title{\Huge CoMP-Enhanced Flexible Functional Split for Mixed Services in Beyond 5G Wireless Networks}

\author{\small \IEEEauthorblockN{Li-Hsiang Shen, Yung-Ting Huang, and Kai-Ten Feng}\\
\IEEEauthorblockA{Department of Electronics and Electrical Engineering \\
National Yang Ming Chiao Tung University, Hsinchu, Taiwan\\ 
gp3xu4vu6.cm04g@nctu.edu.tw, tina344124@gmail.com, and ktfeng@nycu.edu.tw}}


\maketitle

\begin{abstract}
	With explosively escalating service demands, beyond fifth generation (B5G) aims to realize various requirements for multi-service networks, i.e., higher performance of mixed enhanced mobile broadband (eMBB) and ultra-reliable low-latency communication (URLLC) services than 5G. To flexibly serve diverse traffic, various functional split options (FSOs) are specified by 5G protocols enabling different network functions. In order to improve signal qualities for edge users, we consider FSO-based coordinated multi-point (CoMP) transmission as a prominent technique capable of supporting high traffic demands. However, due to conventional confined hardware processing capability, a processor sharing (PS) model is introduced to deal with high latency for multi-service FSO-based networks. Therefore, it becomes essential to assign CoMP-enhanced functional split modes under PS model. A more tractable FSO-based network in terms of ergodic rate and reliability is derived by stochastic geometry approach. Moreover, we have proposed CoMP-enhanced functional split mode allocation (CFSMA) scheme to adaptively assign FSOs to provide enhanced mixed throughput and latency-aware services. The simulation results have validated analytical derivation and demonstrated that the proposed CFSMA scheme optimizes system spectrum efficiency while guaranteeing stringent latency requirement. The proposed CFSMA scheme with the designed PS FFS-CoMP system outperforms the benchmarks of conventional FCFS scheduling, non-FSO network, fixed FSOs, and limited available FSO selections in open literature.

\end{abstract}


\section{Introduction} \label{ch:Intro}

	
	With proliferation of prosperous beyond fifth generation (B5G) networks, there are urgent and diverse service demands requiring high-rate and low-latency system performances along with high capability, flexibility and dynamicality \cite{acm}. As specified by 3rd generation partnership project (3GPP) specifications, there exist three mandatory services consisting of enhanced mobile broadband (eMBB), ultra-reliable low-latency communication (URLLC) and massive machine-type communications (mMTC) \cite{multiservice}, which target on supporting high data rate, low-latency and massive number of connections for devices, respectively. Conventional 5G networks mainly aim to optimize a single service, which is independently conducted under individual data traffic. However, due to limited network resources and growing demands from enormous devices, it becomes compellingly imperative to consider mixed services such that the network can simultaneously provide multiple services while guaranteeing their corresponding requirements. 
	
\begin{figure*}[!t]
		\centering
		\includegraphics[width=5 in]{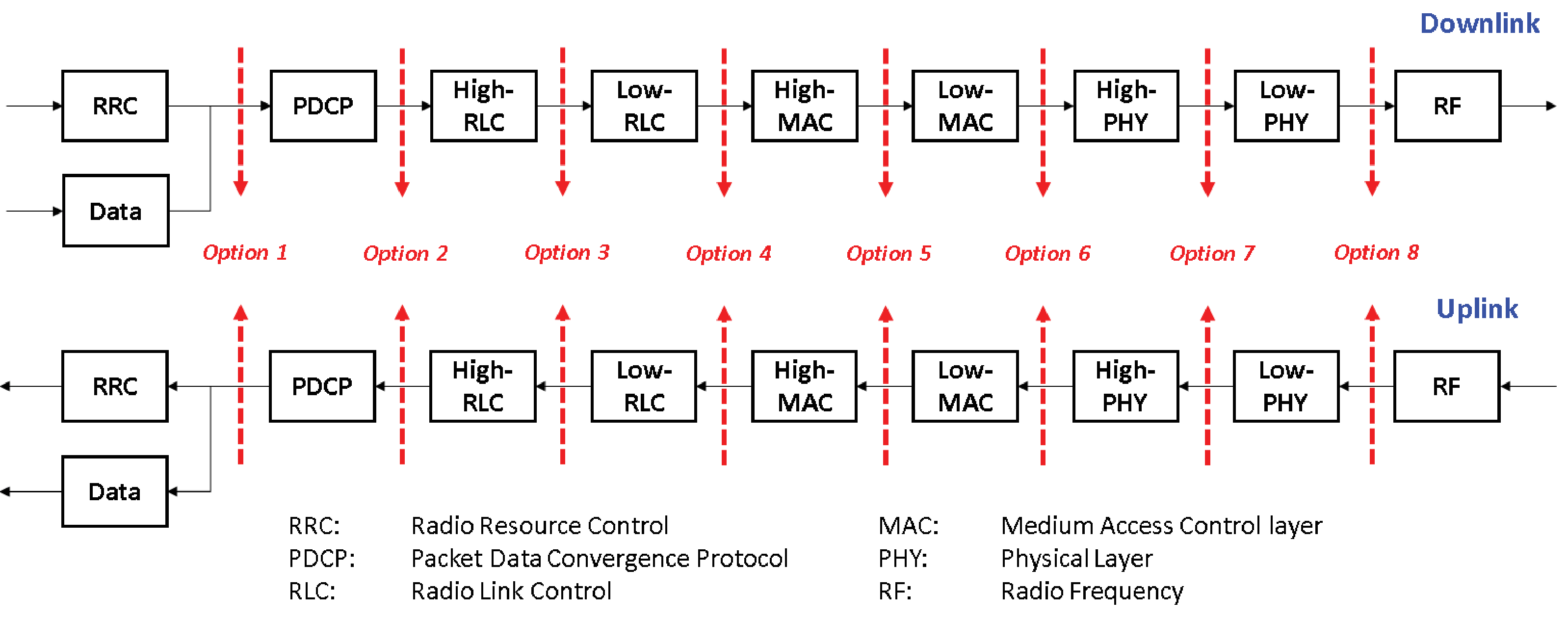}
		\caption{\small The architecture of functional split options including separation of radio resource control (RRC), packet data convergence protocol (PDCP), high-/low-radio link control (RLC), high-/low-MAC (H-/L-MAC), high-/low-PHY (H-/L-PHY) and radio frequency (RF) layers \cite{GPP, tutorial}. RF takes care for antenna port mapping, up-/down-conversion and sampling. PHY is responsible for in-phase/quadrature-phase (IQ) symbol processing, beamforming, resource element mapper, precoding/decoding as well as modulation coding schemes. MAC includes scheduler, multiplexing, channel data transfer, and retransmission mechanism. RLC aims at segmentation and buffering for transmission and retransmission. PDCP consists of ciphering, header compression and numbering. RRC supports numerous functions of system information connection control, measurement configuration and inter-radio mobility.}
		\label{fig:FS}
	\end{figure*}

	In order to support simultaneous services, it potentially requires different framework designs tailored for aggregated and separated network function layers in either centralized or distributed architecture \cite{service_tailored}. For example, joint transmission for eMBB traffic needs coordination functions among base stations (BSs), whilst queuing processing for URLLC services should be performed under efficient and reliable network functions. Therefore, functional split (FS) is designed and specified in 5G to achieve the essential needs for future service in B5G networks, which can be regarded as the revolutionary cloud radio access network \cite{GPP}. Several functional split options (FSOs) over network layers standarized by 3GPP are illustrated in an introductive manner in Fig. \ref{fig:FS}, which contains from higher network resource control layers to lower medium access control (MAC) and physical (PHY) accessing layers \cite{GPP}. The conventional BS implementing full functions has the capability to process global network information \cite{c_ran}. However, it is disadvantageous to deploy fully-functional BSs for a certain service due to its costly infrastructure and high bandwidth requirement. Accordingly, conventional architecture of BSs is partitioned into a central unit (CU) connected to several distributed units (DUs) taking the responsibility of different network functions in order to enhance the flexibility of network services. Centralized schemes tend to adaptively select partitioned options at lower layers, i.e., more functions will be executed at CU, whilst DU only takes care of simple signal processing. On the other hand, in a distributed scenario with partition at higher layers, the DUs perform comparably more network functions than those adopting centralized fashion. The authors of \cite{tutorial} have provided an overview of eight network FSOs supported as shown in Fig. \ref{fig:FS}. It states that FSOs from 5 to 8 can perform coordinated multi-point (CoMP) based joint transmission, retransmission and interference management for eMBB users. However, not all users need CoMP techniques for supporting high data rate demands, which requires dynamic adjustments for different services such as URLLC traffic demands.

	 As for an flexible functional split (FFS) system, the work load of the server is another significant factor since it is related to system operability. When the packet arrival rate is sufficiently high, i.e., high work load, it potentially provokes either performance unsatisfaction or network crash due to explosive buffers. Accordingly, it becomes important to design a novel scheduling scheme and dynamic functional split mode in order to efficiently and effectively share the network resources for a multi-service network. A table for acronyms is established in Table \ref{abbr}. The main contributions of this work are summarized as follows.

\begin{table}
\begin{center}
\scriptsize	
\setstretch{1.1}
\caption {Acronyms}
    \begin{tabular}{|l|l|}
    \hline
        Acronym & Definition \\        
        \hline\hline
		B5G	&	beyond fifth generation	\\ \hline
BS	&	base station	\\ \hline
CCDF	&	complementary cumulative distribution function	\\ \hline
CFSMA	&	CoMP-enhanced functional split mode allocation	\\ \hline
CM/DM	&	centralized/distributed mode	\\ \hline
CoMP	&	coordinated multi-point	\\ \hline
CPU	&	central processing unit	\\ \hline
CU/DU	&	central/distributed unit	\\ \hline
E2E	&	end-to-end	\\ \hline
eMBB	&	enhanced mobile broadband	\\ \hline
FCFS	&	first-come-first-serve	\\ \hline
FFS	&	flexible functional split	\\ \hline
FSO	&	functional split option	\\ \hline
MAC	&	medium access control	\\ \hline
MAR	&	mode allocation ratio	\\ \hline
mMTC	&	massive machine-type communications	\\ \hline
PHY	&	physical layer	\\ \hline
PLD	&	power level difference	\\ \hline
PS	&	processor sharing	\\ \hline
SINR	&	signal-to-interference-plus-noise ratio	\\ \hline
TTI	&	transmission time interval	\\ \hline
URLLC	&	ultra-reliable low-latency communication	\\ \hline

		\end{tabular} \label{abbr}
\end{center}
\end{table}

	\begin{itemize}
		\item We conceive an FFS framework for mixed services of throughput-oriented as well as reliability-aware applications, i.e., eMBB/URLLC, including centralized mode (CM) and distributed mode (DM), as proposed and depicted in Fig. \ref{fig:ar}. As for CM, flexible FSO is conducted at lower layers, where the FFS system enables CoMP-based joint transmissions for rate-oriented eMBB services. On the other hand, DM enabling non-CoMP transmission performs higher layer FSO than CM, where DUs have more network functions leading to potentially lower latency for URLLC services. Under the multi-service network, we have designed a resilient CoMP-enhanced FFS architecture to flexibly assign both CM and DM modes to improve mixed-service quality.
	
		\item A tractable CoMP-enhanced FFS architecture is analyzed by using stochastic geometry approach in terms of signal-to-interference-plus-noise ratio (SINR) and system ergodic rate. Based on theoretical analysis, we can determine the activation of CoMP-based joint transmission via distance-based pathloss from the DUs to associated users. As for URLLC services, we conceive a novel processor sharing (PS) scheduling scheme capable of simultaneously processing traffic data, which possesses lower latency than conventional first-come-first-serve (FCFS) method. Also, with considerations of channel/signal fluctuation, a general end-to-end (E2E) delay model is applied including CU/DU processing using PS, re-transmissions and CoMP overhead. The corresponding reliability for URLLC under the proposed PS-FFS-CoMP framework is also considered relevant to flexible FSO assignment of CM/DM.
		
		\item We have proposed a CoMP-enhanced functional split mode allocation (CFSMA) scheme to flexibly assign FSOs in CM/DM modes for mixed services, which maximizes the ergodic rate of eMBB traffic while guaranteeing URLLC reliability and system operability. The proposed CFSMA in PS and CoMP-enhanced FFS network is proved to be able to select the candidate set of CM/DM mode selection according to reliability and operability requirements. Based on the selected feasible set, it is also proven that CFSMA searches the optimal solution of power level difference (PLD) thresholds which is equivalent to assigning flexible CM/DM modes for multi-services.
		
		\item The performance evaluations validate that our analysis of proposed PS-FFS-CoMP framework asymptotically approaches the simulation results. With higher control overhead, it potentially provokes lower system ergodic rate and larger outage due to higher delay. Moreover, under compellingly high traffic loads, the system utilizing conventional FCFS method cannot fulfill reliability of URLLC users. Benefited from simultaneous processing for multi-service traffic data, the proposed PS-FFS-CoMP network is capable of offering resilient and efficient service for URLLC users while providing more computing and communication resources for improving eMBB users, which outperforms the benchmarks of conventional FCFS scheduling, non-FSO network, fixed FSOs, and limited available FSO selections in open literature.
\end{itemize}

 	The rest of this paper is organized as follows. In Section \ref{Liter}, we investigate open literature for functional split and mixed services. In Section \ref{SYS_MOD}, we provide the system model for functional split mode, PS, and CoMP transmissions. Tractable analysis and  problem formulation for CoMP-Enhanced FSS for mixed-services is demonstrated in Section \ref{TracAna}. The proposed CFSMA scheme is described in Section \ref{CHAPTER_ALGORITHM}. Performance evaluations of proposed CFSMA scheme are conducted in Section \ref{SIM_ANA}. Finally, conclusions are drawn in Section \ref{SEC_CON}.

	\begin{table*}[!t]
		\centering
		\scriptsize
		\caption {\small Comparison of Related Work}
		\begin{tabular}{lcccccccccccccc}
			\hline
			& \cite{service_ns} & \cite{GPP, user_slicing_fs} & \cite{flex_fs} &\cite{adap_fs1} &\cite{delay_constrained_fs} & \cite{PHY_fs} & \cite{multi_objective_fs} &\cite{sliced_ran_fs}  &\cite{adap_fs2} &\cite{power_control_fs} &\cite{capacity_constrained_fs} & \cite{r1, r2, r3} & \cite{r4} & This work \\ \hline \hline
			Functional Split & & \checkmark & \checkmark & \checkmark & & & \checkmark & \checkmark & \checkmark & \checkmark & \checkmark & & \checkmark & \checkmark \\ \hline
			Multi-Services & \checkmark & \checkmark & \checkmark & & & & & & & & & \checkmark & \checkmark & \checkmark \\ \hline
			eMBB Traffic & \checkmark & \checkmark & & & & & & & & & \checkmark & \checkmark & \checkmark & \checkmark \\ \hline
			URLLC Traffic & \checkmark & \checkmark & & & & & \checkmark & \checkmark  & & & & \checkmark & \checkmark & \checkmark \\ \hline
			CoMP & & & & & & & & & & & \checkmark  & & & \checkmark \\ \hline
			Packet Processing & & & & & \checkmark & & & \checkmark & & & & & \checkmark & \checkmark \\ \hline
			Rate Metric & & & \checkmark & \checkmark & & \checkmark & & & & \checkmark & & \checkmark & \checkmark & \checkmark \\ \hline
			Latency Metric & & & \checkmark & \checkmark & \checkmark & & \checkmark & \checkmark & \checkmark & & & \checkmark & \checkmark & \checkmark \\
			\hline
		\end{tabular} \label{RefComparison}
	\end{table*}

\section{Literature Review} \label{Liter}
 	In paper of \cite{service_ns}, the authors have designed a new network function split method by using node selection for different requirements in a multi-service system but with option fixed, which confines the flexibility of network services. Therefore, it becomes compellingly imperative to design flexible functional split options to efficiently utilize network resources. The work in \cite{GPP} has introduced different selections for granularity of CU/DU FSOs, i.e., CU- and DU-oriented and user-centric architecture, which possesses higher adaptivity and flexibility of multi-services for practical implementations. The work in \cite{user_slicing_fs} has implemented a user-centric FSO selection via network slicing that jointly optimizes radio, processing and wireless link resources. With the revolutionary network virtualization and software defined networks, user-centric functional split becomes possible so that multi-service packets can be multiplexed and forwarded to different users with distinct requirements \cite{flex_fs}, which has also been corroborated in \cite{adap_fs1} through experiments on an open platform.
	
	The paper in \cite{delay_constrained_fs} has investigated latency performance for URLLC-based traffic under different FSOs, which considers both processing time and work load of CU and DUs. The authors of \cite{PHY_fs} have studied different functional splits at PHY layer which aims to minimize the operation cost under certain throughput constraints. However, it is implausible to perform dynamic FSO mechanism leading to inflexibility of multi-service traffic. With network slicing in FFS, the work of \cite{multi_objective_fs} has formulated a virtual network embedding problem guaranteed by individual traffic requirements. However, multi-service requirements cannot be simultaneously fulfilled due to the fixed FSO mechanism. The paper \cite{sliced_ran_fs} has designed a joint network slicing and FSO framework targeting at maximizing centralization degree among fronthual and backhaul links between the CU and DUs. Furthermore, \cite{adap_fs2} has developed an adaptive radio access network that is capable of switching between two FSO architectures considering time-varying traffic and bandwidth-limited fronthaul. The results in \cite{power_control_fs} have revealed that the throughput can be substantially enhanced with FFS instead of fixed FSO mechanism. Moreover, the paper of \cite{capacity_constrained_fs} has compared two FSOs and stated that the mixture of FSOs can achieve power preservation considering limited fronthaul capacity.

	However, none of the aforementioned literature of \cite{multi_objective_fs, sliced_ran_fs, adap_fs1,adap_fs2, delay_constrained_fs, power_control_fs, capacity_constrained_fs, PHY_fs, r1, r2, r3, r4} considers flexible FSOs for multi-services, such as hybrid eMBB/URLLC services. Most of existing literature considers a single performance metric of either throughput for eMBB or latency for URLLC users, which are impractical and inapplicable in a real multi-service system. In papers \cite{r1, r2}, the authors have respectively proposed machine learning and robust methods to conduct resource slicing in order to solve coexisting eMBB-URLLC services optimizing both rate- and latency-aware users. A fairness problem for eMBB users while guaranteeing URLLC traffic requirement is designed in \cite{r3}. Though optimizing joint rate-latency performances, \cite{r1, r2, r3} did not consider flexible functional split mechanism. In \cite{r4}, traffic assignment is additionally taken into account, which can be regarded as a type of flexible splitting. To further enhance the user traffic demands, CoMP transmission can be leveraged into the network. Additionally, FFS-enabled CoMP-URLLC transmission has not been discussed in open literature, where joint throughput and end-to-end delay should be optimized, which is also not considered in our previous works of \cite{100,101}. Moreover, multi-DU and multiuser services over variant channels potentially deteriorates the system performance, which should be considered in FFS design. As for an FFS system, the work load of the server is another significant factor since it is related to system operability, which is not considered in existing works as well as our previous work in \cite{102}. The comparison of related works mentioned above is summarized in Table \ref{RefComparison}.

\section{System Model} \label{SYS_MOD}

 	\begin{figure*}[!t]
		\centering
		\includegraphics[width=5.5 in]{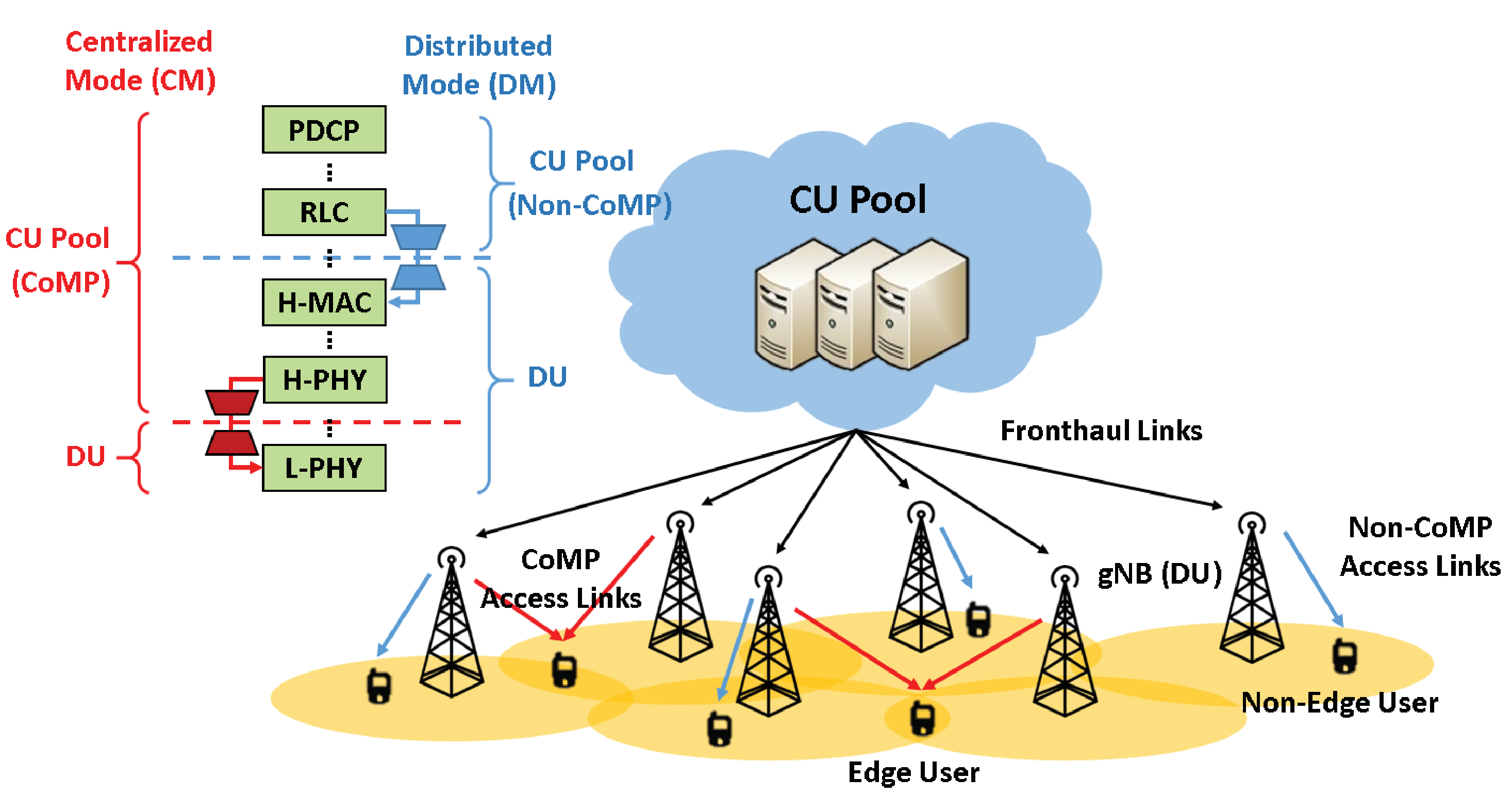}
		\caption{ System architecture of flexible functional split including CM and DM mechanisms. Note that CoMP transmission in CM is enabled via FSO partition at higher/lower PHY, whereas non-CoMP in DM is conducted via FSO separation between RLC and MAC layers.}
		\label{fig:ar}
	\end{figure*}

	\begin{figure} [t]
		\centering
		\subfigure[]
		{\includegraphics[width=2.7 in]{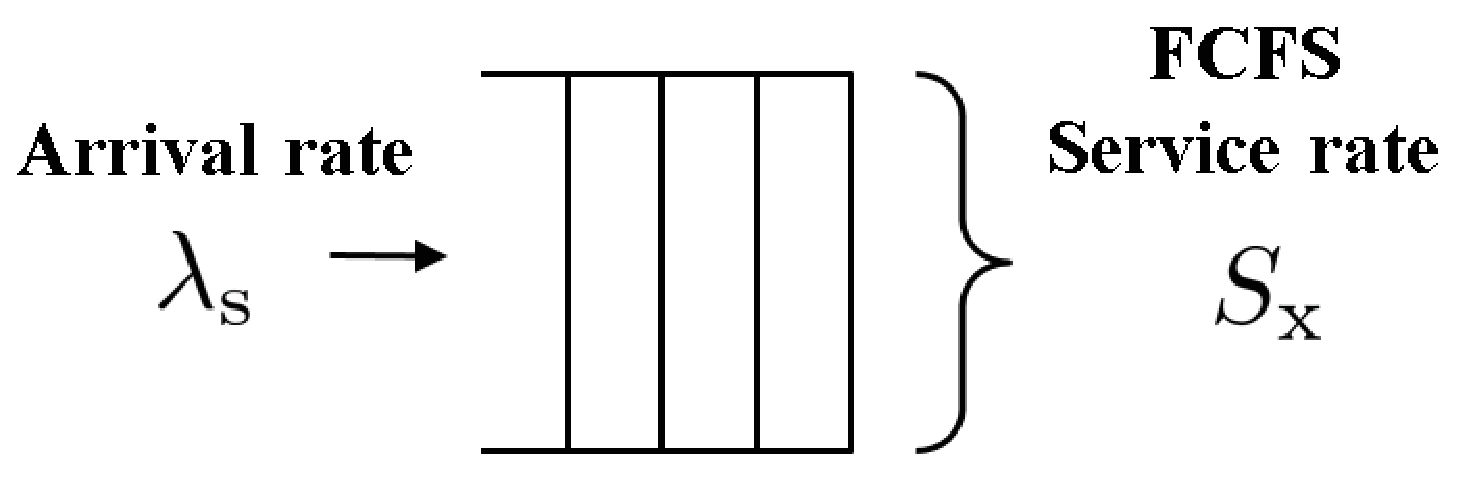} \label{fig:fcfs}}
		\subfigure[]
		{\includegraphics[width=2.7 in]{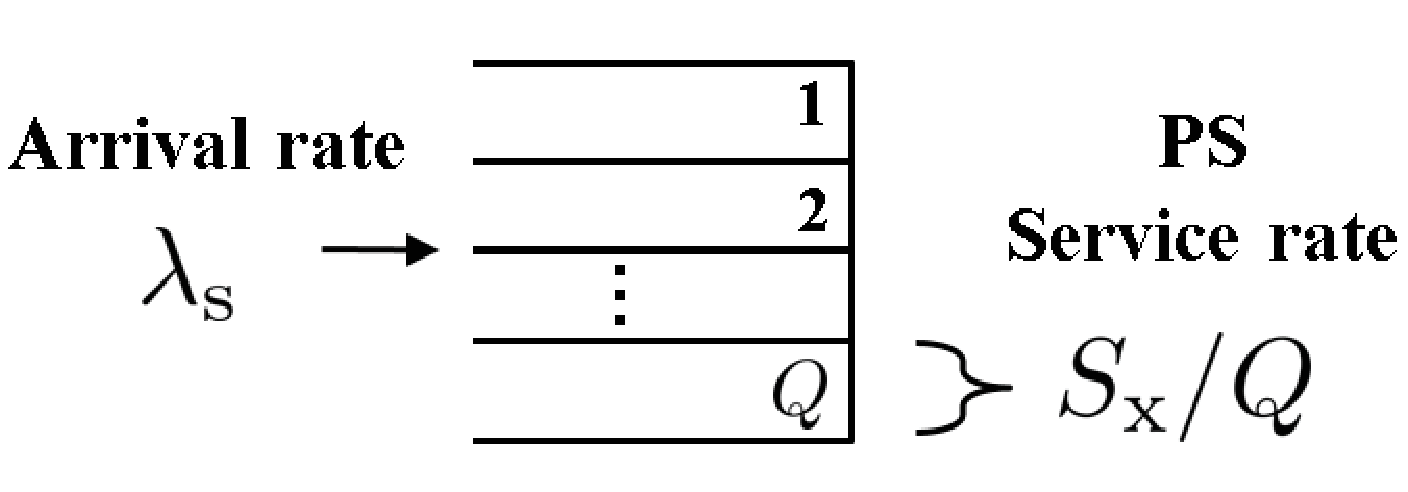} \label{fig:PS}}
		\caption {(a) Conventional FCFS scheduling method; and (b) Proposed scheduling using PS mechanism.} 
		\label{fig:ps}
	\end{figure}

\begin{table}
\begin{center}
\scriptsize	
\setstretch{1.1}
\caption {Definition of System Parameters}
    \begin{tabular}{|l|l|}
    \hline
        Parameters & Notation \\        
        \hline\hline
		Number of DUs & $N$ \\ \hline
		Density of DU and users & $\Lambda_{\D}$, $\Lambda_{\s}$ \\ \hline
		Processing server CU/DU &  $\X \in \{\C,\D\}$ \\ \hline
		eMBB/URLLC service & $\s\in \{\e,\U\}$ \\ \hline
		Serving CM/DM mode & $\M \in \{ \CM, \DM\}$ \\ \hline
		Packet arrival rate and service rate & $\lambda_{\s}$, $S_{\X}$ \\ \hline
		Number of arriving packets & $Q$ \\ \hline
		Portion of processing load & $\beta_{\M,\X}$ \\ \hline
		Distance between DU/user and its set & $r_{n}$, $\mathbf{r}$ \\ \hline
		Joint probability of $c$ nearest DUs & $b(\mathbf{r},r_{c+1})$ \\ \hline
		Pathloss exponent & $\alpha$ \\ \hline
		PLD value and PLD threshold & $v_{n}$, $\boldsymbol\gamma=\{\gamma_{\e},\gamma_{\U}\}$ \\ \hline
		Connected/Allowable number of CoMP DUs & $c$, $C$ \\ \hline
		Auxiliary PLD parameter & $B(\gamma_{\s})$ \\ \hline
		FFS MAR in mode $\M$ & $p_{\M,\s}$ \\ \hline
		SINR for CoMP transmission & $\SINR(c, \mathbf{r},\gamma_{\s})$ \\ \hline
		Small-scale fading and its mean & $h_c$, $\mu$ \\ \hline
		Total interference &  $I_{r_c}$ \\ \hline
		Noise power & $\sigma^2$ \\ \hline
		Serving coverage probability & $p_{cov}\left(T,c,\mathbf{r},\gamma_{\s}\right)$ \\ \hline
		Threshold of CCDF of SINR & $T$ \\ \hline
		Laplacian of interference $I_{r_1}$ & $\mathcal{L}_{I_{r_1}}(\mu T r_1^{\alpha})$ \\ \hline
		Auxiliary parameter for Laplacian & $\nu$ \\ \hline
		Ergodic rate for service $\s$ & $R(\s)$ \\ \hline
		Sojourn time & $\tau_{\X}(\M,\s| Q=q)$ \\ \hline
		Average sojourn time & $\tau_{\X}(\M,\s)$ \\ \hline
		Processing cycles for a packet & $l_{\s}$ \\ \hline
		Work load of server $\X$ & $\rho_{\X}$ \\ \hline
		Total users of service $\s$ & $K_{\s}$ \\ \hline
		Average serving users per DU & $\kappa_{\M}$ \\ \hline
		Serving time & $\tau_{\X}^{ser}(\M,\s)$ \\ \hline
		One-way delay & $t_{\rm{ow}}(\M,\s)$ \\ \hline
		Transmission time interval & $t_{\tti}(\s)$ \\ \hline
		Control overhead in mode $\M$ & $t_{\co}(\M)$ \\ \hline
		End-to-end delay and requirement & $t_{\ee}(\M,\s)$, $t_{\U}^{(max)}$ \\ \hline
		Duration of retransmission & $t_{\re}$ \\ \hline
		Number of failure transmissions & $N_{f}$ \\ \hline
		Maximum allowable retransmission & $N_{f,\M,\s}^{(max)}$ \\ \hline
		Successful transmission probability & $p_{succ}(\M,\s)$ \\ \hline
		Reliability requirement & $R(\e)$ \\ \hline
		PLD set for reliability/operability & $\boldsymbol{ \gamma}^{\mathcal{R}}$, $\boldsymbol{ \gamma}^{\mathcal{O}}$ \\ \hline	
		Feasible PLD candidate set and its optimum & $\boldsymbol{\gamma}^{\mathcal{F}}$, $\boldsymbol{\gamma}^{*}$ \\ \hline
		Auxiliary parameters for theorems & $\mathcal{A}_{\M}$, $\mathcal{S}_1$, $\mathcal{S}_2$, $G(x)$ \\ \hline
		\end{tabular} \label{Parameters}
\end{center}
\end{table}

\subsection{Functional Split Architecture}

	As shown in Fig. \ref{fig:ar}, we consider an FFS network $\Psi$ consisting of a single CU, $N$ DUs serving multiusers. Both DUs and users follow the distribution of Poisson point process (PPP) with densities denoted as $\Lambda_{\D}$ and $\Lambda_{\s}$, respectively. Note that poor channel quality occurs at edge users, and vice versa for non-edge users. We define the processing server as $\X \in \{\C,\D\}$ for CU and DU with specific service $\s\in \{\e,\U\}$ for throughput-oriented eMBB and reiliabity and latency-aware URLLC traffic, respectively. The packet arrival rate of users with service $\s$ is denoted as $\lambda_{\s}$. We consider that the transmission of desirable data packets are experienced through wired fronthaul between the CU and DUs and access links from DUs to serving users. Furthermore, a hybrid service of mixed throughput-/latency-aware traffic is considered with these two services possessing substantially different network requirements, i.e., eMBB is for rate-oriented services, whilst URLLC aims at low-latency and time-critical communications\textsuperscript{\ref{note0}}\footnotetext[1]{Based on the specification of 5G new radio \cite{num1}, time-critical communications can be supported under microwave with numerology types 1 and 2. We can use higher-order numerology in above-6 GHz to meet the requirement of diverse applications, which however induces further issues, such as narrow-beam training and beam misalignment requiring additional complex modeling and designs. \label{note0}}. Without loss of generality, the packet size of eMBB service is comparably larger than that of URLLC traffic. As shown in Fig. \ref{fig:ps}, we demonstrate the conventional processing and the proposed PS mechanism inspired by \cite{PS}. When adopting conventional FCFS scheduling method as depicted in Fig. \ref{fig:fcfs}, latency requirement is potentially not satisfied due to longer queuing waiting time. Furthermore, congestion of eMBB packets may be induced by prioritized latency-aware packets. Accordingly, to deal with above-mentioned problem, we consider a PS scheduling scheme which is a preferable solution to simultaneously achieve stringent rate and latency requirements. As shown in Fig. \ref{fig:PS}, the main concept of PS is to equally assign the segmented queue capacity to different data streams, where the service rate is shared to all packets in the server. That is, arriving data enters the service immediately with no line to await. We define the service rate at FCFS as $S_{\X}$ for $\X$ processing server, therefore the service rate in PS is divided by $Q$ arriving packets resulting in a processing rate of $S_{\X}/Q$. Accordingly, the experienced queuing delay of PS can be significantly mitigated compared to that of FCFS scheduling. In other words, latency-aware packets under PS shall experience much smaller queueing delay since it will be processed instantly. 

	In the considered FFS network, two functional split modes are supported, i.e., CM and DM, where the discrepancy is whether layers from high-PHY to high-MAC are deployed at either CU or DU as shown in Fig. \ref{fig:ar}. We denote the serving mode $\M \in \{ \CM, \DM\}$ for CM and DM, respectively. As illustrated in Fig. \ref{fig:ar}, CM enables CoMP transmission by partitioning FSO at high- and low-PHY layers, whilst DM for non-CoMP separates its network functions at RLC and MAC layers. Referring from 3GPP specification TR 38.801 \cite{GPP}, due to the fact that high-PHY is in control of signal precoding selection, the users can thus be served by CoMP based joint transmission in CM. By contrast, DUs in DM cannot serve CoMP services since both MAC and PHY layers are deployed at DUs, that is, users can only be served by a single DU. We consider that packets experience the same process in both CM and DM modes from the CU to DUs. Therefore, the following equality holds
	\begin{equation}
	\beta_{\M,\C}+\beta_{\M,\D}=1, \forall \M, \label{proload}
	\end{equation}
	where $\beta_{\M,\X}$ denotes the portion of processing load allocated to server $\X$ in mode $\M \in \{ \CM, \DM\}$, and $\beta_{\M,\X} \in (0,1)$. Note that when packets are operated in CM, the CU will have comparably higher processing load than that in DUs, i.e., the inequality $\beta_{\CM,\C}>\beta_{\DM,\C}$ is satisfied. A symbol table is listed in Table \ref{Parameters}.
	

\subsection{CoMP-Enhanced FFS}

	As mentioned previously, in CM for FFS network, DUs are capable of serving users with impoverished channels using CoMP transmissions, i.e., signal quality can be substantially improved. Note that each DU can flexibly decide the mode to meet individual service requirement. We consider a PLD method to determine the necessity of CoMP transmissions \cite{joint}, which implicitly indicates the quality of corresponding channel. We assume that the serving DU will serve its nearest users, however, other non-serving DUs will provoke interferences to desirable users deteriorating rate performance. Therefore, the central CU will convert interferences from the other DUs into CoMP based signals if the estimated received power of PLD from the other DUs are sufficiently high to menace desired transmission. In other words, the user with low signal quality will be served by more than one DU enabling CoMP mechanism. We define $\mathcal{C}$ as the serving DU set of the reference user. We consider that all DUs are sorted and indexed based on increments of distance from DUs to reference user, and therefore the first element of $\mathcal{C}$ denotes the nearest DU providing the strongest signal and acting as the mandatory serving DU. Moreover, the PLD between the nearest DU and DU $n$ is defined as $\upsilon_{n} =\frac{{r}_{1}^{-\alpha}}{r_{n}^{-\alpha}}$, where $\alpha$ is the pathloss exponent and $r_n$ is the distance between the reference user and the $n$-th nearest DU. Note that inspired from stochastic geometry theory under randomly distributed users, small-scale channel fading (appropriate for instantaneous transmission model) is averaged, whilst only distance-based large-scale fading will significantly influence the channel quality \cite{pld_comp} and accordingly system performance. Based on PLD method, the strategy of DU $n\in \mathcal{C}$ to serve the reference user under service $\s$ can be determined as
\begin{equation}
\D n \in \mathcal{C}, \text{if } \upsilon_{n} \leq \gamma_{\s},
\label{eq:du_assignment}
\end{equation}
where $\gamma_{\s}$ is PLD threshold of service $\s$. Notice that PLD threshold $\gamma_{\s}$ plays an important role to decide the feasible serving DU set and operating mode for users. We consider that the reference user can associate with $c$ DUs, where $c$ satisfies $c \leq C \leq N$ with $C$ indicating the maximum allowable number of candidate CoMP DUs. Therefore, the probability that reference user connects to $c$ DUs at the same time can be computed as
\begin{align}
\pr\left(|\mathcal{C}|=c\right) 
&= \underbrace{\int_{r_1=0}^\infty \int_{r_2=r_1}^{B(\gamma_{\s})r_1} ... \int_{r_{c}=r_{c-1}}^{B(\gamma_{\s})r_1}}_{\text{Candidate CoMP DU}} \underbrace{\int_{r_{c+1}=r_c}^\infty}_{\text{Non-CoMP DU}} 
 b(\mathbf{r},r_{c+1}) dr_{c+1}dr_{c}...dr_1 
\notag \\ &\overset{(a)}{=} {\int_{r_1=0}^\infty \int_{r_2=r_1}^{B(\gamma_{\s})r_1} ... \int_{r_{c}=r_{c-1}}^{B(\gamma_{\s})r_1}} {\int_{r_{c+1}=B(\gamma_{\s})r_1}^\infty}  b(\mathbf{r},r_{c+1}) dr_{c+1}dr_{c}...dr_1,
\label{eq:c_du}
\end{align}
where the notation $|\cdot|$ denotes cardinality of the given set. Notation $B(\gamma_{\s})=\gamma_{\s}^{\frac{1}{\alpha}}$ is obtained from the PLD threshold in $\eqref{eq:du_assignment}$ indicating whether DU $n$ is selected as one of elements in the candidate CoMP set. In $\eqref{eq:c_du}$, it includes candidate CoMP and non-CoMP DUs which vary in upper limits of integrals. The equality (a) holds since the lower limit of the last integral for non-CoMP DU should satisfy the PLD constraint, i.e., the non-CoMP DU is out of range of CoMP service. Consider the example illustrated in Fig. \ref{fig:gamma2FSM}, for users in CM, the secondary nearest DU satisfying the PLD condition of $\eqref{eq:du_assignment}$ can be selected as one of CoMP candidates. On the other hand, DM user is not served by CoMP transmission due to unsatisfaction of PLD threshold. Moreover, the joint probability of $c$ nearest DUs can be obtained from  \cite{distance} as
\begin{equation}
	b(\mathbf{r}) = b(r_1,r_2,...,r_{c}) = e^{-\pi\Lambda_{\D} \sum_{i=1}^{c} r_{i}^{2}} \cdot \left(2\pi\Lambda_{\D}\right)^{c} \cdot \prod_{i=1}^{c} r_{i},
\label{eq:PPP}
\end{equation}
where $\mathbf{r}$ represents the set of $c$ nearest DUs deployed at the distance from $r_1$ to $r_{c}$, i.e., $\mathbf{r}= \{r_{i}, \forall 1\leq i \leq c\}$. Note that probability of $b(\mathbf{r},r_{c+1})$ in $\eqref{eq:c_du}$ can be obtained in a similar manner in $\eqref{eq:PPP}$ with additional term of $r_{c+1}$. Thus, according to $\eqref{eq:c_du}$, for arbitrary user belonging to service $\s$, we can define the FFS mode allocation ratio (MAR) in CM and DM respectively as
\begin{subequations}
	\begin{align}
&p_{\CM,\s}=\sum_{c=2}^C \pr\left(|\mathcal{C}|=c\right), \\
&p_{\DM,\s}=\pr\left( |\mathcal{C}|=1\right).
	\end{align} \label{eq:C is 2}
\end{subequations}
We can observe from $\eqref{eq:C is 2}$ that PLD threshold $\gamma_{\s}$ potentially affects the number of CoMP-enabled DUs, which alternatively influences the processing functionality of CU and DUs in different modes of CM and DM. Therefore, it becomes compellingly imperative to determine the optimal PLD to meet the requirement of multi-services in a CoMP-enhanced FFS network.

	\begin{figure}
	\centering
	\includegraphics[width=4in]{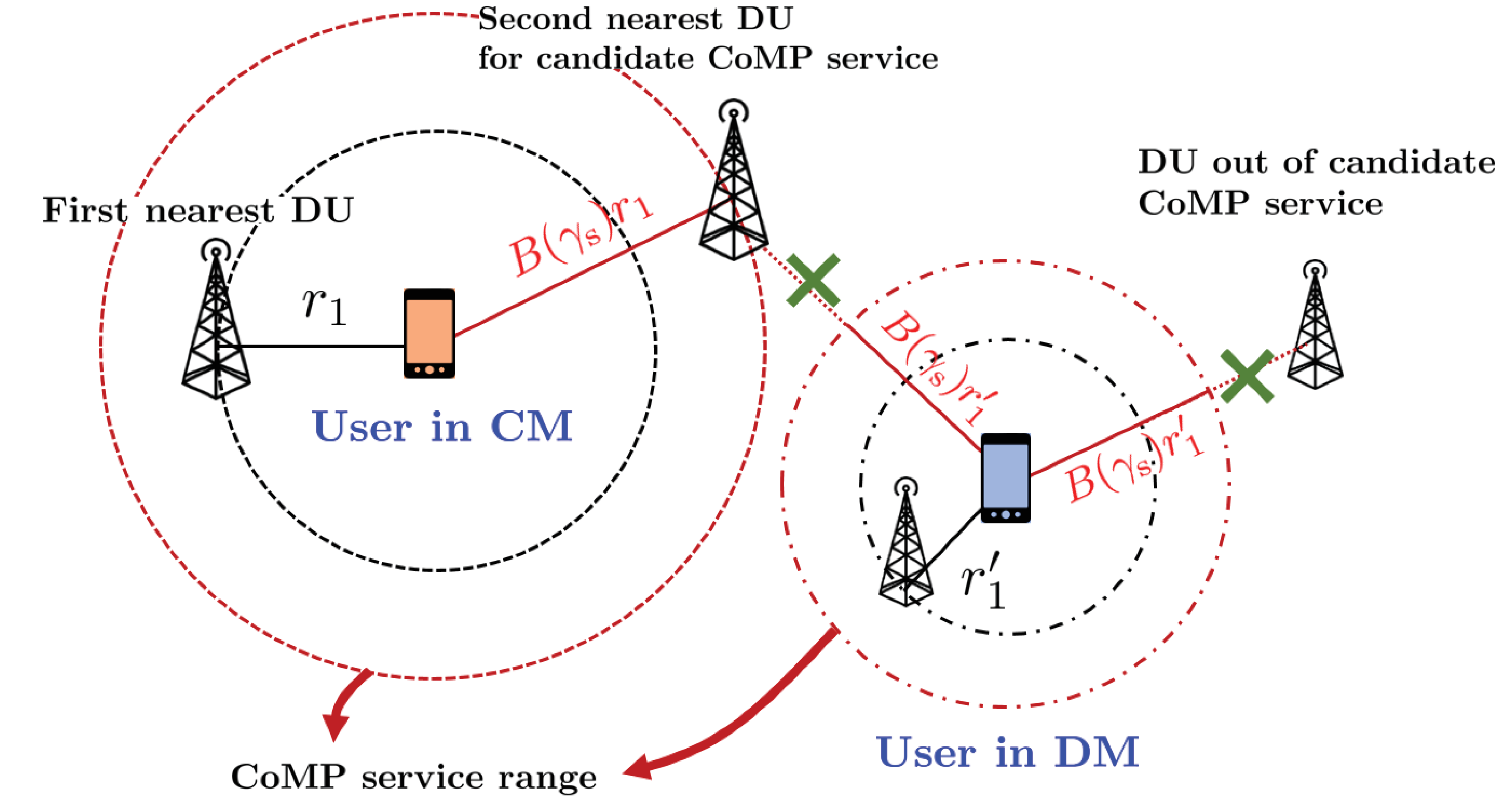}
	\caption{Illustration of FFS CM/DM mode selection. For the user in CM, the second nearest DU is clustered as the candidate CoMP service which should satisfy PLD threshold in CoMP service with range as $B(\gamma_{\s})r_{1}$. On the other hand, the second and third nearest DUs out of CoMP service range with $B(\gamma_{\s})r'_{1}$ will not serve the user, which becomes DM mode.}
	\label{fig:gamma2FSM}
\end{figure}

\section{Tractable Analysis for CoMP-Enhanced FSS for Mixed Services} \label{TracAna}

\subsection{SINR and Throughput Analysis}

We consider a wireless scenario of omnidirecitonal downlink transmission. Accordingly, for the reference user connected to $c$ DUs located at distances from $r_1$ to $r_c$ with PLD threshold $\gamma_{\s}$, the SINR for CoMP transmissions is defined as
\begin{equation}
\SINR(c, \mathbf{r},\gamma_{\s}) = \frac{h_1G_{1}r_1^{-\alpha}+...+h_{c}G_{c}r_{c}^{-\alpha}}{I_{r_{c}}+\sigma^2},
\label{eq:sinr}
\end{equation}
where $h_{c}$ is the small-scale Rayleigh channel fading of the $c$-th nearest DU following exponential distribution with its mean of $\mu$. $G_{c}$ is the antenna gain of DU $c$. The notation $I_{r_{c}}=\sum_{c'=c+1}^{C} h_{c'}G_{c'}r_{c'}^{-\alpha}$ is total interference from the non-serving DUs outside the distance $r_{c}$. Note that equal power and an equivalent constant antenna gain from DUs is considered for derivation simplicity, where $G_c$ and $G_{c'}$ are neglected in the following derivations\textsuperscript{\ref{note1}}\footnotetext[2]{We consider equal power as a baseline case compared to power control based method. In our FFS system, it would require further complex scheme when coupled rate and latency factors are considered, which can be extended as a future work inspired by the power control methods under a simplified network in \cite{r1,r5}. Moreover, though we have neglected antenna gain model here due to its constant value consideration, this work can be extended to directional beamforming in either sub-6 or above-6 GHz transmissions. It requires additional complex beam modeling and analysis under different antenna patterns and beamforming techniques \cite{my1} as future inspiring research issues. \label{note1}}. $\sigma^2$ is the normalized noise power. We consider identical transmit power of DUs in $\eqref{eq:sinr}$ which is normalized within $\sigma^{2}$. Also, the transmission SINR model is generally used in both analyses of (a) eMBB throughput and (b) URLLC delay and reliability performances in the following sections. Therefore, the serving coverage probability can be defined as the complementary cumulative distribution function (CCDF) of SINR which is given by
\begin{equation}
p_{cov}\left(T,c,\mathbf{r},\gamma_{\s}\right)= \pr\left(\SINR(c,\mathbf{r}, \gamma_{\s})\geq T\right),
\label{eq:pc_general}
\end{equation}
where $\SINR\left(c, \mathbf{r}, \gamma_{\s}\right)$ is the SINR from $c$ DUs with PLD threshold $\gamma_{\s}$, and $T$ indicates the threshold of CCDF of SINR. We can therefore observe from $\eqref{eq:pc_general}$ that coverage probability is relevant to CoMP parameters with respect to the number of serving DUs and PLD threshold. Based on conditional probability property, we first analyze SINR in $\eqref{eq:sinr}$ under given $\gamma_{\s}$ which is represented by $\eqref{eq:pc_condition}$ as shown at top of this page.
\begin{figure*}[!t]
\begin{align}
	&\pr\Big(\SINR(c, \mathbf{r},\gamma_{\s}) \geq T\Big) = \pr\left( \frac{h_1r_1^{-\alpha}+...+h_{c}r_{c}^{-\alpha}}{I_{r_{c}}+\sigma^2}\geq T \right) \notag\\
	&\overset{(a)}{=} \int_{T'r_1^\alpha}^{\infty} \pr \left( h_1 = h_1'\right) \cdot \mathbbm{E}_{I_{r_{c}}} \left[ \pr \left( \sum_{i=2}^c h_{i}r_{i}^{-\alpha} \geq T'-h_1r_1^{-\alpha} \Big| h_1=h_1'\right) \right] dh_1' \notag\\
	&\overset{(b)}{=}  \int_{T'r_1^\alpha}^{\infty} \! \pr\left( h_1 \! = \! h_1' \right) \int_{\left( T' \!-\! h_1r_1^{-\alpha}\right) r_2^{\alpha}}^{\infty} \! \pr \left( h_2 \! = \! h_2'\right) \! \cdot \mathbbm{E}_{I_{r_{c}}} \left[ \pr \left( \sum_{i=3}^c h_{i}r_{i}^{-\alpha} \! \geq \! T'\!-\!\sum_{j=1}^2 h_j' r_{j}^{-\alpha} \right) \right] dh_2'dh_1' \notag\\
	& \overset{(c)}{=} \int_{T'r_1^\alpha}^{\infty} \int_{\left(T'-h_1r_1^{-\alpha}\right)r_2^{\alpha}}^{\infty} ... \int_{ \! \left( \! T' - \sum\limits_{i=1}^{c-1} h_{i}r_{i}^{-\alpha} \! \right) \! r_{c}^{\alpha}}^{\infty} \prod_{j=1}^c \pr\left( h_{j} \!=\! h_{j}'\right) dh_{c}'...dh_1'
\label{eq:pc_condition}
\end{align}
\hrulefill
\end{figure*}
In $\eqref{eq:pc_condition}$, (a) holds due to the independence of small scale channel fading with auxiliary parameter of $T'=T(I_{r_{c}}+\sigma^2)$. Based on marginal probability property, the equality (a) is derived indicating the integral of channel distribution of the first nearest DU. Therefore, in a similar transformation method, the inner integral of equality (b) implies the integral of the second nearest DU given the first one. Accordingly, a general expression of CCDF of SINR via equality (c) is derived for $c$ CoMP serving DUs. As mentioned previously in $\eqref{eq:c_du}$, the probability of the number of connected DUs is relevant to the PLD threshold $\gamma_{\s}$, i.e., distance of $\mathbf{r}=\{r_{i} | \forall 1\leq i\leq c\}$ is related to $\gamma_{\s}$. Thus, combining joint distribution of $\eqref{eq:c_du}$ and $\eqref{eq:pc_condition}$ yields a general form $\eqref{eq:pc_general}$ for $c$ FFS-based CoMP DUs as
\begin{equation}
\begin{aligned}
p_{cov}(T,c,\mathbf{r},\gamma_{\s}) =
\int_{r_1=0}^\infty \int_{r_2=r_1}^{B(\gamma_{\s})r_1} ... \int_{r_{c}=r_{c-1}}^{B(\gamma_{\s})r_1} \int_{r_{c+1}=B(\gamma_{\s})r_1}^\infty \pr\Big( \SINR(c, \mathbf{r},\gamma_{\s}) \!\geq \! T \Big)
b(\mathbf{r},r_{c+1}) dr_{c+1}dr_{c}...dr_1.
\end{aligned}
\label{eq:pc_gamma}
\end{equation}
Due to difficulties to infer substantial insights from $\eqref{eq:pc_gamma}$, we consider an implementable case for FFS-based CoMP, i.e., maximum number of $C=2$ DUs to serve a single user\textsuperscript{\ref{note2}}\footnotetext[3]{As for CoMP transmission in CM, there induces overhead of synchronization and channel information exchange. The mechanism becomes more complex with comparably higher overhead with the increment of participating DUs, which also leads to no closed-form solution for theoretical analysis under mixed services. However, it is capable of providing a closed-form solution when $C=2$, whilst the more flexible network with $C>2$ is an open issue as future works \label{note2}}. Note that according to the definition of FFS network, we only have CoMP transmission for mode CM and conventional single link for non-CoMP service in DM. Therefore, we can have the following propositions for DM and CM, respectively.

\begin{prop} \label{prop1}
For an FFS-enabled network, the serving coverage probability as CCDF of SINR in the DM of non-CoMP transmission can be derived as
\begin{align} \label{cov1}
	p_{cov}(T,1,\mathbf{r},\gamma_{\s}) &= \int_{r_1=0}^\infty \int_{r_2=B(\gamma_{\s})r_1}^\infty e^{-\mu T r_1^\alpha \sigma^2} \cdot \mathcal{L}_{I_{r_1}}(\mu T r_1^\alpha) b(r_1,r_2) dr_2dr_1,
\end{align}
where $\mathcal{L}_{I_{r_1}}(\mu T r_1^{\alpha})=\mathbbm{E}_{I_{r_1}}\left[ e^{-\mu T r_1^{\alpha} I_{r_1}}\right]$ is Laplace transform of interference $I_{r_1}$ for the first nearest DU.
\end{prop}
\begin{proof}
For a non-CoMP transmission, we substitute $c=1$ in $\eqref{eq:pc_gamma}$ to obtain the corresponding coverage probability as
\begin{align}
	p_{cov}(T,1,\mathbf{r},\gamma_{\s}) 
	&= \int_{r_1=0}^\infty \int_{r_2=B(\gamma_{\s})r_1}^\infty \pr\Big(\Gamma(1, \mathbf{r},\gamma_{\s}) \geq T\Big) b(r_1,r_2) dr_2dr_1 \notag\\
	&= \int_{r_1\!=\!0}^\infty \int_{r_2\!=\!B(\gamma_{\s})r_1}^\infty \pr\Big( h_1r_1^{-\alpha} \!\geq\! T(I_{r_{1}}\!+\!\sigma^{2})\Big) b(r_1,r_2) dr_2dr_1.
\label{eq:pc_special1}
\end{align}
By considering the Rayleigh fading channel with exponential distribution, we can further acquire the following equation as
\begin{align} \label{eq:pc_special11}
	\pr\Big( h_1r_1^{-\alpha} \geq T(I_{r_{1}}+\sigma^{2})\Big) &\overset{\underset{\mathrm{(a)}}{}}{=} \mathbbm{E}_{I_{r_{1}}} \left[ e^{-\mu r_{1}^{\alpha}T(I_{r_{1}}+\sigma^{2})}\right] \overset{\underset{\mathrm{(b)}}{}}{=} e^{-\mu T r_1^\alpha \sigma^2} \mathcal{L}_{I_{r_1}}(\mu T r_1^\alpha),
\end{align}
where (a) follows the exponential distribution of $h_{1}$ with the mean of $\mu$, and (b) defines the Laplace transform of $\mathcal{L}_{I_{r_1}}(x)=\mathbbm{E}_{I_{r_1}}\left[ e^{-xI_{r_1}}\right]$ of $I_{r_1}$ for the first nearest DU \cite{stochastic}. We can obtain the result of $\eqref{cov1}$ by substituting $\eqref{eq:pc_special11}$ into $\eqref{eq:pc_special1}$ which completes the proof.
\end{proof}

\begin{prop} \label{prop2}
For an FFS-based CoMP transmission, the serving coverage probability as CCDF of SINR in CM by considering $c=2$ CoMP DUs can be obtained as
\begin{align} \label{cov2}
	p_{cov}(T,2,\mathbf{r},\gamma_{\s}) \!=\! \int_{r_1=0}^\infty \int_{r_2=r_1}^{B(\gamma_{\s})r_1} \left[ \left(1\!+\!\nu \right)e^{-\mu T r_1^\alpha \sigma^2} \mathcal{L}_{I_{r_2}}(\mu T r_1^\alpha)  \!-\!\nu e^{-\mu T r_2^\alpha \sigma^2} \mathcal{L}_{I_{r_2}}(\mu T r_2^\alpha) \right] b(r_1,r_2) dr_2 dr_1,
\end{align}
where $\nu=\frac{1}{(\frac{r_1}{r_2})^{-\alpha}-1}$, and $\mathcal{L}_{I_{r_2}}(\mu T r_1^\alpha)$ and $\mathcal{L}_{I_{r_2}}(\mu T r_2^\alpha)$ are Laplace transform expressions of $I_{r_2}$ for the second nearest DU.
\end{prop}
\begin{proof} 
For CoMP transmissions with $c=2$ DUs, the coverage probability can be written as
\begin{align}
	p_{cov}(T,2,\mathbf{r},\gamma_{\s})  
	&= \int_{r_1=0}^\infty \int_{r_2=r_1}^{B(\gamma_{\s})r_1} \pr\Big(\Gamma(2, \mathbf{r},\gamma_{\s}) \geq T\Big) b(r_1,r_2) dr_2dr_1 \notag\\
	&= \int_{r_1=0}^\infty \int_{r_2=r_1}^{B(\gamma_{\s})r_1} \pr\Big( h_1r_1^{-\alpha} + h_2r_2^{-\alpha} \geq T(I_{r_{2}}+\sigma^{2})\Big) b(r_1,r_2) dr_2dr_1.
\label{eq:pc_special2}
\end{align}
Based on Rayleigh fading channel, we can derive the followings of CoMP transmissions as
\begingroup
\allowdisplaybreaks
\begin{align} \label{eq:pc_special22}
	&\pr\Big( h_1r_1^{-\alpha} + h_2r_2^{-\alpha} \geq T(I_{r_{2}}+\sigma^{2})\Big)  \notag\\
	&\overset{\underset{\mathrm{(a)}}{}}{=} \mathbbm{E}_{I_{r_{2}}} \left[ \int_{T(I_{r_{2}}+\sigma^2)r_1^\alpha}^{\infty} \int_{\left(T(I_{r_{2}}+\sigma^2)-h_1r_1^{-\alpha}\right)r_2^{\alpha}}^{\infty}  \pr\left( h_{1} \!=\! h_{1}'\right) \pr\left( h_{2} \!=\! h_{2}'\right) dh_{1}'dh_{2}' \right] \notag\\
	&\overset{\underset{\mathrm{(b)}}{}}{=}  \mathbbm{E}_{I_{r_{2}}} \left[ \left(1+\nu\right) e^{-\mu T r_{1}^{\alpha}\left( I_{r_{2}} + \sigma^{2} \right)}  - \nu e^{-\mu T r_{2}^{\alpha}\left( I_{r_{2}} + \sigma^{2} \right)} \right] \notag \\
	&\overset{\underset{\mathrm{(c)}}{}}{=} \left(1+\nu\right)e^{-\mu T r_1^\alpha \sigma^2} \mathcal{L}_{I_{r_2}}(\mu T r_1^\alpha) -\nu e^{-\mu T r_2^\alpha \sigma^2} \mathcal{L}_{I_{r_2}}(\mu T r_2^\alpha),
\end{align}
\endgroup
where (a) follows the generalized form of $\eqref{eq:pc_condition}$ with $c=2$. The equality (b) is the result of double integrals with respect to $h_{1}'$ and $h_{2}'$ along with  notational definition of $\nu=\frac{1}{(\frac{r_1}{r_2})^{-\alpha}-1}$. The last equality (c) holds due to Laplace transform of $\mathcal{L}_{I_{r_2}}(x)=\mathbbm{E}_{I_{r_2}}\left[ e^{-xI_{r_2}}\right]$ of $I_{r_2}$ for the second nearest DU. We can obtain the coverage probability of CoMP transmission of $c=2$ in $\eqref{cov2}$ by substituting $\eqref{eq:pc_special22}$ into $\eqref{eq:pc_special2}$, which completes the proof.
\end{proof}
\noindent We know that $B(\gamma_{\s})$ is a monotonic increasing function with respect to $\gamma_{\s}$. Therefore, as can be observed from results of $\eqref{cov1}$ in Proposition \ref{prop1} and of $\eqref{cov2}$ in Proposition \ref{prop2}, we can infer that PLD threshold substantially affects the integrals in their lower and upper limits, respectively, which accordingly influences coverage probability. We have concluded this phenomenon in the following corollary.
\begin{corollary} \label{cc}
The coverage probability of DM in $\eqref{cov1}$ is monotonically decreasing with the increment of PLD threshold of $\gamma_{\s}$. While, the coverage probability of CM in $\eqref{cov2}$ is increasing with larger $\gamma_{\s}$, which strikes a compelling tradeoff between CM/DM adjustment.
\end{corollary}
\noindent Moreover, we can acquire a general expression of achievable ergodic rate considering arbitrary service $\s$ under both CM and DM for an FFS-enhanced CoMP transmission which is given by
\begingroup
\allowdisplaybreaks
\begin{align}
	R(\s) & = \Emean \left[ \sum_{c=1}^{C} \ln\Big(1+\SINR(c,\mathbf{r}, \gamma_{\s})\Big) \right]  \overset{(a)}{=}\int_0^\infty \sum_{c=1}^{C} \pr \Big[\ln\Big( 1+\SINR(c,\mathbf{r}, \gamma_{\s}) \Big) \geq \xi \Big] d\xi \notag\\
	& \overset{(b)}{=} \int_0^\infty \sum_{c=1}^C \pr\Big(\SINR(c,\mathbf{r}, \gamma_{\s}) \geq e^{\xi}-1 \Big) d\xi \overset{(c)}{=}\int_0^\infty \sum_{c=1}^C \frac{\pr\Big( \SINR(c,\mathbf{r}, \gamma_{\s}) \geq \theta \Big)}{\theta+1} d\theta, \label{eq:se}
\end{align}
\endgroup
where (a) holds due to the conversion from the expectation to the integral of CCDF of the given variables. In equality (b), we transform the ergodic rate into a form of CCDF of SINR by leaving $\Gamma (c,\mathbf{r}, \gamma_{\s}) $ at the left hand side of the inequality. Moreover, equality (c) is obtained by changing variable as $\theta = e^{\xi}-1$. As a result, we can obtain the general ergodic rate by substituting coverage probabilities of $\eqref{cov1}$ and $\eqref{cov2}$ via $\eqref{eq:pc_general}$ into $\eqref{eq:se}$ for arbitrary services.

\subsection{Delay Analysis}

	For FFS-based latency-aware applications, we provide the tractable analysis for the delay time which is defined from the time when a packet arrives in CU to that when the packet is correctly received at users. Therefore, for the considered FFS network, a data packet will experience the sojourn time both in CU and DUs, propagation delay over fronthaul and access links, and overhead of control signalling of FFS functions. Note that the sojourn time is relevant to the data traffic status in the server, i.e., the system is potentially congested due to explosive number of data packets which provokes long sojourn time. Given $Q=q$ packets existed in the server, the sojourn time at server $\X \in \{\C,\D\}$ for service $\s\in \{\e,\U\}$ in mode $\M\in \{\CM, \DM\}$ can be obtained as \cite{queueing}
\begin{equation}
	\tau_{\X}(\M,\s| Q=q)=\frac{\beta_{\M,\X} l_{\s}(q+1)}{S_\X},
\end{equation}
where $\beta_{\M,\X}$ is the portion of processing load defined in $\eqref{proload}$, $l_{\s}$ indicates the required processing cycles for a packet of service ${\s}$, and $S_{\X}$ denotes the processing cycle rate of central processing unit (CPU) at server $\X\in \{\C, \D\}$. Note that $l_{\s}$ reveals that a packet experiences FFS network layers before transmission at access links in terms of the processing cycles required both in CU and DUs. The probability that general sojourn time $\tau_{\X}(\M,\s)$ can be approximated as
\begin{align}
	\pr\Big( \tau_{\X}(\M,\s)=\tau_{\X}(\M,\s|Q=q)\Big) &= \pr\left( Q= q\right)=(\rho_{\X})^{q}\cdot(1- \rho_{\X}), \label{probtau}
\end{align}
which implies that $q$ arrived packets are being processed while the arriving one is awaiting. Note that the equivalence is derived due to the asymptotic meaning of time duration and total number of packets. In $\eqref{probtau}$, notation $\rho_{\X}$ denotes the work load of server $\X$, which is given by
\begin{equation} \label{rho}
\rho_{\X}=\left\{  
\begin{aligned}  
&\frac{\sum_{\s} \sum_{\M} K_{\s}\lambda_{\s}p_{\M,\s}\beta_{\M,\C} l_{\s} }{S_{\C}},& & \text{if } \X = \C ,\\  
&\frac{\sum_{\s} \sum_{\M} \kappa_{\M} K_{\s}\lambda_{\s}p_{\M,\s}\beta_{\M,\D} l_{\s}  }{S_{\D}\cdot N},& & \text{if } \X = \D, \\ 
\end{aligned}  
\right.
\end{equation}
where $K_{\s}$ is the total users of service $\s$, $\lambda_{\s}$ indicates the packet arrival rates for a certain service and $\kappa_{\M}$ is defined as the average serving users per DU considering total $N$ DUs in the system. In $\eqref{rho}$, we consider FFS MAR $p_{\M,\s}$ for service $\s$ under mode $\M$, which potentially assigns different processing loads for CU $\beta_{\M,\C}$ and DUs $\beta_{\M,\D}$ in either CM or DM. For example, higher $p_{\CM, \s}$ in CM means that more processing work loads of CoMP-enabled functions is performed for service $\s$. Under conventional FCFS scheduling, each data packet will acquire full processing resources with the serving time denoted as $\tau_{\X}^{ser}(\M,\s)=\frac{\beta_{\M,\X}l_{\s}}{S_{\X}}$. Therefore, for the proposed PS scheduling mechanism, the average sojourn time at server $\X$ for service $\s$ in mode $\M$ can be expressed by
\begin{equation} \label{Eman}
	t_{\X}(\M,\s)=\Emean_q \left[\tau_{\X}(\M,\s | Q=q)\right]=\frac{\tau_{\X}^{ser}(\M,\s)}{1-\rho_{\X}},
\end{equation}
which implies that the shared service rate of PS scheduling is reduced to $S_{\X}(1-\rho_{\X})$. In the FFS network, we consider that fronthaul is established via wired optical fiber links where delay can be neglected. Accordingly, based on $\eqref{Eman}$, the average one-way delay of service $\s$ in mode $\M$ is obtained as
\begin{equation}
	t_{\rm{ow}}(\M,\s)=t_{\C}(\M,\s)+t_{\D}(\M,\s)+t_{\tti}(\s)+ t_{\co}(\M), \label{delay}
\end{equation}
where $t_{\C}(\M,\s)$ and $t_{\D}(\M,\s)$ are derived from $\eqref{Eman}$ containing processing cycles and loads. Notation of $t_{\tti}(\s)$ is the transmission time interval of service $\s$ over access links from DUs to users, which is composed of consecutive transmission symbols in time domain in a specific transmit direction as specified in \cite{GPP2}. The control overhead in mode $\M$ is denoted as $t_{\co}(\M)$, which intuitively results in $t_{\co}(\CM)\geq t_{\co}(\DM)$ since CoMP operation in CM requires higher signalling overhead than non-CoMP transmission of DM. Without loss of generality, we consider identical transmission time interval (TTI) for all supported services. To characterize the system delay performance, we consider E2E delay defined as the interval when a packet arrives in CU to that when the user successfully receives and decodes it. In a practical system, retransmission mechanism potentially takes place due to either unsuccessful transmission or reception. Therefore, the E2E delay of mode $\M$ and service $\s$ after $N_{f}$ failure transmissions is defined as
\begin{equation} \label{EEDelay}
	t_{\ee}(\M,\s)=\left(N_{f}+1\right) \cdot t_{\rm{ow}}(\M,\s)+N_{f} \cdot t_{\re},
\end{equation}
where $t_{\re}$ is the duration of retransmission operation. Note that the identical time of $t_{\re}$ is considered for both CM and DM modes due to sufficient wired fronthaul bandwidth. We denote $t_{\ee}(\s)$ as the average system E2E delay of service $\s$, and therefore the reliability of service $\s$ can be expressed as \cite{multiservice}
\begin{align} \label{eq:reliability}
\pr\Big( t_{\ee}(\s)\leq t_{\s}^{(max)} \Big)
&= \sum_{\M} \pr\Big( t_{\ee}(\M,\s)\leq t_{\s}^{(max)} \Big) p_{\M,\s} \overset{(a)}{=} \sum_{\M} \pr\Big( N_{f}\leq N_{f,\M,\s}^{(max)} \Big) p_{\M,\s},
\end{align}
where $t_{\s}^{(max)}$ is the maximum allowable delay of service $\s$.
The equality of (a) holds by substituting $\eqref{EEDelay}$, where $N_{f,\M,\s}^{(max)}$ indicates the maximum allowable attempts of retransmissions for service $\s$ in mode $\M$ which is derived as
\begin{equation}
	N_{f,\M,\s}^{(max)}=\Big\lfloor \frac{( t_{\s}^{(max)} - t_{\rm{ow}}(\M,\s) )^+}{t_{\rm{ow}}(\M,\s)+t_{\re}}\Big\rfloor,
\label{eq:max_f}
\end{equation}
where $\lfloor \cdot \rfloor$ is the ceiling operation and $\left( x \right)^+ = \max \left(0,x\right)$. Moreover, the CCDF of retransmission attempts under maximum allowable $N_{f,\M,\s}^{(max)}$ trials is given by
\begin{align}
	\pr \Big( N_{f} \leq N_{f,\M,\s}^{(max)} \Big) &= \sum_{i=0}^{N_{f,\M,\s}^{(max)}} p_{succ}(\M,\s) \cdot \Big(1-p_{succ}(\M,\s) \Big)^i = 1-\Big( 1 - p_{succ}(\M,\s) \Big)^{N_{f,\M,\s}^{(max)}+1},
\label{eq:reliability_m}
\end{align}
where $p_{succ}(\M,\s) = p_{cov}(T_{\s},c,\mathbf{r},\gamma_{\s})$ indicates the successful transmission probability which is equivalent to the coverage probability in $\eqref{eq:pc_general}$ under given SINR threshold $T_{\s}$ for service $\s$, and $c=\{1,2\}$ for DM and CM, respectively. We can infer from $\eqref{eq:reliability_m}$ that the retransmission will be initiated if the required signal quality is lower than the given threshold. However, more transmission retrials induce higher E2E delay which potentially deteriorates the reliability performance. Therefore, it becomes compellingly essential to allocate appropriate work loads $\rho_{\X}$ and FFS mode $p_{\M,\s}$ to satisfy multi-service requirements. Moreover, allocating FFS mode $p_{\M,\s}$ is equivalent to determine the PLD threshold of $\boldsymbol\gamma=\{\gamma_{\e},\gamma_{\U}\}$, i.e., CoMP transmission with higher SINR will be performed in most DUs if $p_{\CM,\s} \geq p_{\DM,\s}$ meaning that more CoMP-based work loads are assigned, which strikes a potential balance between ergodic rate and reliability performances. To elaborate a little further, inspired by \cite{i1,i2,i3}, we consider an average manner in our tractable analysis of hybrid services due to difficulty to acquire the exact number of queuing packets from a large-scale network. Moreover, the successful transmission probability of latency-aware service is partially affected by the coverage probability coming from SINR in CM/DM mechanisms. Therefore, with serving users uniformly randomly distributed, it requires probabilistic analysis to observe how this novel system works. As an potential extension of this work, instantaneous resource allocation and scheduling is treated as a considerably complex mechanism which can be designed with conceptual priority-based processor sharing and dynamic adjustment of CM/DM modes.

\subsection{Problem Formulation for Mixed Services} \label{PromFor}
	In the FFS-enabled network, we consider mixed services possessing the tradeoff among system ergodic rate, delay and reliability. We can observe from $\eqref{eq:du_assignment}$ that larger PLD threshold will cluster more candidate CoMP DUs to serve eMBB users, which implies that a large portion of users is served under CM. Accordingly, improving the ergodic rate and probability of successful transmission become critical for eMBB services. However, if the system serves more eMBB users in CM mode using CoMP, the CU will conduct higher work loads which potentially induces longer queuing time resulting in lower reliability for latency-aware users. If the latency-aware users are considered as higher priority, some throughput-oriented users will probably be served by a single DU, i.e., CoMP is not activated in DM which deteriorates throughput performance. Therefore, it becomes compellingly imperative to design the optimal FFS mechanism to strike the balance for mixed services. The proposed problem formulation is represented by
\begin{subequations} \label{ProbFor}
	\begin{align}
	& \max_{\boldsymbol\gamma=\{\gamma_{\e},\gamma_{\U}\}} \quad R(\e) \label{eq:obj}\\ 
	& \text{s.t.} \  \pr\Big( t_{\ee}(\U) \leq t_{\U}^{(max)}\Big) \geq \delta_{\U} \label{eq:st1}\\
	&\qquad  0 < \rho_{\X}< 1, \quad\ \forall \X \in \{\C,\D\}, \label{eq:st2}\\
	&\qquad \sum_{\M} p_{\M,\s}=1, \quad \forall \s \in \{\e, \U \}, \label{eq:st3}\\
	& \qquad 0\leq p_{\M,\s}\leq 1, \quad \forall \M\in\{\CM,\DM\},\s\in\{\e,\U\}. \label{eq:st4}
	\end{align}
\end{subequations}
In problem $\eqref{ProbFor}$, we aim at maximizing the ergodic rate $R(\e)$ obtained from $\eqref{eq:se}$ of throughput-oriented services while guaranteeing the reliability of latency-aware service acquired from $\eqref{eq:reliability}$ to be larger than the required threshold of latency-aware users $\delta_{\U}$ in $\eqref{eq:st1}$. Note that adjusting PLD $\boldsymbol\gamma = \{\gamma_{\e},\gamma_{\U}\}$ is equivalent to determining the FFS modes which potentially influence MAR $p_{\M,\s}$, ergodic rate and reliability performances. The constraint of $\eqref{eq:st2}$ stands for the operability of FFS system in CU and DUs meaning that total work loads are operable within the server capability. Moreover, constraints of $\eqref{eq:st3}$ and $\eqref{eq:st4}$ guarantee the feasible MAR values for allocating portion of total work loads to different modes of CM and DM. However, it is non-intuitive to acquire the optimal solution of $\boldsymbol \gamma$ in problem $\eqref{ProbFor}$ due to the existence of mixed services in terms of ergodic rate and reliability constraint. In the following, we propose a CFSMA scheme to resolve the optimization problem in $\eqref{ProbFor}$.

\section{Proposed CoMP-Enhanced Functional Split Mode Allocation (CFSMA) Scheme} \label{CHAPTER_ALGORITHM}

In our proposed CFSMA scheme, we first search for the feasible set of $\boldsymbol \gamma =\{\gamma_{\e}, \gamma_{\U}\}$ to fulfill reliability constraints of $\eqref{eq:st1}$ and $\eqref{eq:st2}$. Based on the feasible set, we can then obtain the optimum solution of $\{\gamma_{\e},\gamma_{\U}\}$ in order to maximize ergodic rate of eMBB service. As observed from Fig. \ref{fig:gamma2FSM}, deriving the optimum is equivalent to determining the functional split mode $p_{\M,\s}$ for each user connected to DUs according to PLD threshold $\boldsymbol \gamma$, i.e., $p_{\M,\s}$ is a function of $\boldsymbol\gamma$ derived from $\eqref{eq:c_du}$ and $\eqref{eq:C is 2}$. We denote $\boldsymbol{ \gamma}^{\mathcal{R}}$ and $\boldsymbol{ \gamma}^{\mathcal{O}}$ as reliable and operable sets of $\boldsymbol \gamma$ which imply satisfaction of constraints $\eqref{eq:st1}$ and $\eqref{eq:st2}$, respectively.

\subsection{Candidate Set Selection for Reliability and Operability}
	Based on $\eqref{rho}$, we can infer that $\rho_{\X}$ is monotonic with respect to $p_{\CM,\s}$ given packet arrival rate and network density. Note that we only consider MAR of CM mode $p_{\CM,\s}$ since the summation of both CM and DM modes is one, i.e., $p_{\CM,\s}+p_{\DM,\s}=1$. By jointly solving the following arithmetic problem 
\begin{equation}
\left\{
\begin{aligned}
&0<\rho_{\X}< 1, \quad\forall\X,\\
&\sum_{\M} p_{\M,\s}=1, \quad\forall \s,\\
&p_{\M,\s}\geq 0, \quad\forall \M,\s,\\
\end{aligned}
\right.
\label{operable}
\end{equation}
we can obtain the corresponding upper and lower bounds of candidate solution set of operability for $p_{\CM,\s}$, which are denoted as $p_{\CM,\s}^{\mathcal{O},ub}$ and $p_{\CM,\s}^{\mathcal{O},lb}$, respectively. Therefore, the candidate set of operability $\boldsymbol{ \gamma}^{\mathcal{O}}$ is acquired based on upper and lower bounds of $p_{\CM,\s}^{\mathcal{O},ub}$ and $p_{\CM,\s}^{\mathcal{O},lb}$. Moreover, we can obtain the candidate solution set $\boldsymbol{\gamma}^{\mathcal{R}}$ with reliability constraint $\eqref{eq:st1}$ held based on Theorem \ref{theorem:1}.
\begin{theorem} \label{theorem:1}
	For reliability constraint $\eqref{eq:st1}$ with any service $\s$, the candidate set of reliability $\boldsymbol{\gamma}^{\mathcal{R}}$ is monotonic and continuous.
\end{theorem}
\begin{proof}
	As observed from $t_{\ee}(\M,\s)$ in $\eqref{EEDelay}$, the average sojourn time $t_{\X}(\M,\s)$ in $\eqref{Eman}$ is related to $p_{\CM,\s}$ in $\eqref{rho}$. By taking the second-order derivative of $t_{\X}(\M,\s)$ with respect to $p_{\CM,\s}$, we have $\eqref{dtp}$ as shown at top of this page.
\begin{figure*}[!t]
\begin{align} \label{dtp}
\frac{{d}^2 t_{\X}(\M,\s)}{{ {d} p_{\CM,\s}^2}}=
\left\{\begin{array}{ll}
 	\frac{2 \tau_{\X}^{ser}(\M,\s)}{\left(1-\rho_{\X}\right)^{3}} \left[ \sum_{\s} K_{\s}\lambda_{\s}l_{\s} \left(\beta_{\CM,\X}-\beta_{\DM,\X}\right)\right]^2, &\text{if} \X=\C, \\ \\
 	\frac{2 \tau_{\X}^{ser}(\M,\s)}{\left(1-\rho_{\X}\right)^{3}} \left[ \frac{1}{N}\sum_{\s} K_{\s}\lambda_{\s} l_{\s} \left(\beta_{\CM,\X}\kappa_{\CM}-\beta_{\DM,\X}\kappa_{\DM}\right)\right]^2, &\text{if} \X=\D.
\end{array} \right.
\end{align}
\hrulefill
\end{figure*}
We can observe from $\eqref{dtp}$ that $\frac{{d}^2 t_{\X}(\M,\s)}{{ {d} p_{\CM,\s}^2}}$ is positive since $\tau_{\X}^{ser}(\M,\s)\geq 0$ and $0<\rho_{\X} <1$, which implies that $t_{\X}(\M,\s)$ is a convex function with respect to $p_{\CM,\s}$. Since other parameters in one-way delay $t_{\ow}(\M,\s)$ are constant, both $t_{\C}(\M,\s)$ and $t_{\D}(\M,\s)$ dominate $t_{\ee}(\U)$ leading to a convex function of $t_{\ee}(\U)$ with respect to $p_{\CM,\s}$ \cite{convex}. Based on $\eqref{eq:max_f}$, $N_{f,\M,\s}^{(max)}$ is derived to be a monotonically decreasing function of $t_{\ow}(\M,\s)$. Consequently, $N_{f,\M,\s}^{(max)}$ is a concave function with respect to $p_{\CM,\s}$, and so is probability of retransmission attempts in $\eqref{eq:reliability_m}$. The reliability term of $\eqref{eq:reliability}$ is also a concave function according to the affine function of $p_{\CM,\s}+p_{\DM,\s}=1$. Accordingly, for arbitrary concave function, the solution set enabling $\eqref{eq:st1}$ is almost surely continuous \cite{convex}. Therefore, we can obtain $p_{\CM,\s}^{\mathcal{R},ub}$ and $p_{\CM,\s}^{\mathcal{R},lb}$ as the upper and lower bounds of reliability set, respectively. Based on $\eqref{eq:c_du}$ and $\eqref{eq:C is 2}$, the functional split MAR $p_{\M,\s}$ for considered 2 CoMP DUs becomes
\begin{equation}
	\left\{
	\begin{aligned}
	& p_{\CM,\s}=\int_{r_1=0}^\infty \int_{r_2=r_1}^{B(\gamma_{\s})r_1} b(r_1,r_2) dr_2dr_1,\\
	& p_{\DM,\s}=\int_{r_1=0}^\infty \int_{r_2=B(\gamma_{\s})r_1}^{\infty} b(r_1,r_2) dr_2dr_1.
	\end{aligned}
	\right.
	\label{eq:coverage}
\end{equation}
It can be seen that $p_{\CM,\s}$ monotonically increases with both $B(\gamma_{\s})$ and $\gamma_{\s}$ as the upper limit of integral. As a result, the reliability constraint $\eqref{eq:st1}$ possesses monotonic and continuous candidate set of $\boldsymbol{\gamma}^{\mathcal{R}}$ based on the upper and lower bounds of $p_{\CM,\s}^{\mathcal{R},ub}$ and $p_{\CM,\s}^{\mathcal{R},lb}$ for any service $\s$. This completes the proof.
\end{proof}

According to the sets of $\boldsymbol{\gamma}^{\mathcal{R}}$ and $\boldsymbol{\gamma}^{\mathcal{O}}$ bounded by their upper and lower bounds, the feasible set can be obtained as $\boldsymbol \gamma^{\mathcal{F}}=\boldsymbol{\gamma}^{\mathcal{R}} \cap\boldsymbol{\gamma}^{\mathcal{O}}$. Within the feasible set, any solution $\boldsymbol \gamma \in \boldsymbol{\gamma}^{\mathcal{F}}$ satisfies the constraints in problem $\eqref{ProbFor}$.

\subsection{Ergodic Rate Maximization}
	After deriving the feasible set satisfying both reliability and operability constraints in problem $\eqref{ProbFor}$, we then find the optimal solution $\boldsymbol{\gamma}^{*} \in \boldsymbol{\gamma}^{\mathcal{F}}$ for maximizing the ergodic rate of eMBB services, which is stated in Theorem \ref{theorem:2}.
\begin{theorem}
	The optimal solution of $\boldsymbol{\gamma}^{*} \in \boldsymbol{\gamma}^{\mathcal{F}}$ in problem $\eqref{ProbFor}$ is obtained by selecting $\boldsymbol{\gamma}$ with the largest $\gamma_{\s}$ while satisfying constraints of reliability $\eqref{eq:st1}$ and operability $\eqref{eq:st2}$.
	\label{theorem:2}
\end{theorem}
\begin{proof}
	Based on $\eqref{eq:se}$, ergodic rate can be obtained from coverage probability of non-CoMP and CoMP transmissions derived from Propositions \ref{prop1} and \ref{prop2}, respectively. Though it is difficult to perform convexity derivation from complex double integrals in $\eqref{eq:se}$ and coverage probability, solving the maximization of ergodic rate can be asymptotically equivalent to deriving the total coverage probability in the inner summation of integral of $\eqref{eq:se}$. According to $\eqref{cov1}$ and $\eqref{cov2}$, the total coverage probability is given by
\begin{align} \label{eq:rewrite_pc}
&		p_{cov}(T,1,\mathbf{r},\gamma_{\s}) +p_{cov}(T,2,\mathbf{r},\gamma_{\s}) \notag \\ &=\int_{r_1=0}^\infty e^{-\mu T r_1^\alpha \sigma^2}
		\left( \int_{r_2=B(\gamma_{\s})r_1}^\infty \mathcal{A}_{\DM} b(r_1,r_2) dr_2  +\int_{r_2=r_1}^{B(\gamma_{\s})r_1} \mathcal{A}_{\CM} b(r_1,r_2) dr_2 \right) dr_1,
\end{align}	
where $\mathcal{A}_{\DM}= \mathcal{L}_{I_{r_1}}(\mu T r_1^\alpha)$ and $\mathcal{A}_{\CM}=(1+\nu) \mathcal{L}_{I_{r_2}}(\mu T r_1^\alpha) -\nu e^{-\mu T (r_2^\alpha-r_1^\alpha) \sigma^2} \mathcal{L}_{I_{r_2}}(\mu T r_2^\alpha)$. We can observe from $\eqref{eq:rewrite_pc}$ that 
coverage probability regarding candidate MAR solution of DM is monotonically decreasing with increment of PLD $\gamma_{\s}$ in the lower bound of first integral. However, coverage probability of CM is increasing with larger $\gamma_{\s}$ in the upper limit of second integral. This strikes a potential tradeoff between CoMP and non-CoMP transmissions, which can also be inferred from Corollary \ref{cc}. Thus, we proceed to prove that either $\mathcal{A}_{\DM}$ or $\mathcal{A}_{\CM}$ is dominant, i.e., it is equivalent to proving that either larger or smaller value of $\gamma_{\s}$ can provide maximization of ergodic rate. Comparing two integrands related to DM and CM in $\eqref{eq:rewrite_pc}$, we have
	\begin{align} \label{eq:comparison}
		\mathcal{A}_{\CM}-\mathcal{A}_{\DM} 
		&=(1+\nu) \mathcal{L}_{I_{r_2}}(\mu T r_1^\alpha) -\nu e^{-\mu T (r_2^\alpha-r_1^\alpha) \sigma^2} \mathcal{L}_{I_{r_2}}(\mu T r_2^\alpha)-\mathcal{L}_{I_{r_1}}(\mu T r_1^\alpha)\notag\\
		&=(1+\nu) \left(\mathcal{L}_{I_{r_2}}(\mu T r_1^\alpha)-\mathcal{L}_{I_{r_1}}(\mu T r_1^\alpha)\right)  +\nu \left(\mathcal{L}_{I_{r_1}}(\mu T r_1^\alpha)-e^{-\mu T (r_2^\alpha-r_1^\alpha) \sigma^2} \mathcal{L}_{I_{r_2}}(\mu T r_2^\alpha)\right)\notag\\
		&=(1+\nu) \mathcal{S}_1 +\nu \mathcal{S}_2,
	\end{align}
where 
\begin{align}
	&\mathcal{S}_1\!=\!e^{\Lambda_{\D} \pi \left[ r_2^2 - r_1^2G\left(\left(\frac{r_1}{r_2}\right)^\alpha \right) \right]} \!-\! e^{\Lambda_{\D} \pi r_{1}^2\left(1-G(1)\right)},\\
	& \mathcal{S}_2\!=\!e^{\Lambda_{\D} \pi r_{1}^2\left(1-G(1)\right)} \!-\! e^{\Lambda_{\D} \pi r_{2}^2\left(1-G(1)\right)-\mu T \left(r_2^\alpha-r_1^\alpha\right) \sigma^2}.
\end{align}
Note that we define an auxiliary parameter $G(x)\!=\!\frac{2 \left(\mu T \right)^{\frac{2}{\alpha}}}{\alpha} \int_{0}^{\infty} h^{\frac{2}{\alpha}} \left( \Gamma(-\frac{2}{\alpha},\mu T h x) - \Gamma(-\frac{2}{\alpha}) \right) \mu e^{-\mu h}dh$, where $\Gamma(x)$ is gamma function and $\Gamma(s,x)=\int_{x}^{\infty} t^{s-1}e^{-t} dt$ is upper imcomplete gamma function \cite{stochastic}. We can infer that $G(x)\geq 0$ is monotonically increasing with $x$ and will be almost surely $G(1)\geq 1$ under the reasonable configuration of averaged channel gain, number of deployed DUs and sufficient large threshold. Due to monotonicity property of the exponential function, we know that comparison between parameters of $e^{x_1}$ and $e^{x_2}$ is equal to comparing their power terms of $x_{1}$ and $x_{2}$. Accordingly, we can derive that $\mathcal{S}_1\geq 0$ according to
\begingroup
\allowdisplaybreaks
\begin{align} \label{S1}
		&r_2^2-r_1^2G\left( \left( \frac{r_1}{r_2} \right)^{\alpha} \right)- r_{1}^2 \left(1-G(1) \right) = r_2^2-r_1^2 \left[ 1-G(1)+G \left( \left(\frac{r_1}{r_2}\right)^\alpha \right) \right] \geq 0.
\end{align}
\endgroup
Furthermore, we have $\mathcal{S}_2\geq 0$ since
\begin{equation} \label{S2}
		\Lambda_{\D} \pi\left( 1-G(1) \right) \left( r_1^2-r_2^2 \right)+\mu T \sigma^2 \left(r_2^\alpha-r_1^\alpha\right)\overset{(a)}{\geq}0
\end{equation}
where (a) holds since $r_2\geq r_1$. From $\eqref{S1}$, $\eqref{S2}$ and $\nu \geq 0$, it can be inferred that $\eqref{eq:comparison}$ is positive, i.e., $\mathcal{A}_{\CM} \geq \mathcal{A}_{\DM}$ is proved. Therefore, the second integral in the bracket of $\eqref{eq:rewrite_pc}$ regarding CM-based transmission will dominate the system performance. By selecting larger $\gamma_{\s}$ constrained by reliability $\eqref{eq:st1}$ and operability $\eqref{eq:st2}$, the optimal ergodic rate in problem $\eqref{ProbFor}$ can be achieved. This completes the proof.
	\end{proof}

The concrete steps of proposed CFSMA scheme is demonstrated in Algorithm \ref{alg:FSMA}. We first numerically acquire candidate sets of operability $\boldsymbol{\gamma}^{\mathcal{O}}$ and of reliability $\boldsymbol{\gamma}^{\mathcal{R}}$ based on $\eqref{operable}$ and Theorem \ref{theorem:1}, respectively. By considering joint candidate results, we can obtain the total feasible set as $\boldsymbol{\gamma}^{\mathcal{F}}$. According to Theorem \ref{theorem:2}, we then numerically search the optimal solution $\boldsymbol{\gamma}^{*}$ with the largest $\gamma_{\s} \in \boldsymbol{\gamma}^{\mathcal{F}}$. 
Based on the given $\boldsymbol{\gamma}^{*}$ and $\eqref{eq:du_assignment}$, we can classify users into either CM-based CoMP or DM-based non-CoMP transmissions. The user will be served under CoMP if at least two DUs are deployed within its service range satisfying $\eqref{eq:du_assignment}$, i.e, the FFS mode can be assigned for all users with mixed services. The generic problem possesses a complexity order of $\mathcal{O}\left( |\gamma_{\e}|\cdot |\gamma_{\U}| \cdot |\mathcal{D}_{\CM, \e}| \cdot |\mathcal{D}_{\CM, \U}| \right)$, where $\mathcal{D}_{\CM, \e}$ and $\mathcal{D}_{\CM, \U}$ are defined as the solution sets of $p_{\CM, \e}$ and $p_{\CM, \U}$, respectively. The operation of $|\mathcal{X}|$ is the cardinality of a set $\mathcal{X}$. Note that all parameters are continuous, which require the quantization method for exhaustive search. Based on constraints of feasible MAR values in $\eqref{eq:st3}$ and $\eqref{eq:st4}$, the complexity of problem $\eqref{ProbFor}$ reduces to $\mathcal{O}\left( |\gamma_{\e}|\cdot |\gamma_{\U}| \right)$. Satisfying $\eqref{eq:st1}$ and $\eqref{eq:st2}$ based on solution of $\eqref{operable}$ and Theorem \ref{theorem:1} yields a computational complexity order of $\mathcal{O}\left( |\gamma_{\e}^{\mathcal{F}}| \cdot |\gamma_{\U}^{\mathcal{F}}| \right)$, where $|\gamma_{\e}^{\mathcal{F}}|\leq |\gamma_{\e}|$ and $|\gamma_{\U}^{\mathcal{F}}|\leq |\gamma_{\U}|$. According to derivation of Theorem \ref{theorem:2}, the proposed algorithm CFSMA has a comparably lower complexity of $\mathcal{O}\left( \max\left( |\gamma_{\e}^{\mathcal{F}}|, |\gamma_{\U}^{\mathcal{F}}| \right) \right)$.

\begin{algorithm}[t]
\small
	\caption{Proposed CFSMA Scheme}
	\SetAlgoLined
	\DontPrintSemicolon
	\label{alg:FSMA}
	\begin{algorithmic}[1]
		\STATE \textbf{Initialization:}\\
	
		{ - Reliability requirement $\delta_{\U}$}\\
		{ - E2E delay requirement $t_{\U}^{(max)}$ for URLLC users}\\
		{ - DU density $\Lambda_{\D}$}\\
		{ - Throughput-/Reliability-oriented user density $\Lambda_{\s}$, $\forall \s\in \{\e,\U\}$}\\
		\STATE Obtain the candidate set of operability $\boldsymbol{\gamma}^{\mathcal{O}}$ by solving $\eqref{operable} $\\
		\STATE Acquire the candidate set of reliability $\boldsymbol{\gamma}^{\mathcal{R}}$ based on Theorem \ref{theorem:1}\\
		\STATE Derive feasible set $\boldsymbol{\gamma}^{\mathcal{F}}=\boldsymbol \gamma^{\mathcal{O}} \cap \boldsymbol \gamma^{\mathcal{R}}$\\
		\STATE Find the optimal solution $\boldsymbol \gamma^*$ with the largest $\gamma_{\s}\in \boldsymbol \gamma^{\mathcal{F}}$ based on Theorem \ref{theorem:2} \\
		\STATE According to the obtained optimal threshold $\boldsymbol \gamma^*$, we classify users into either CoMP or non-CoMP transmissions from $\eqref{eq:du_assignment}$ \\
		\STATE CM-based CoMP is performed once the user is served by $|\mathcal{C}| > 1$ DUs; otherwise, DM-based non-CoMP is executed when the user is served by a single DU, i.e, $|\mathcal{C}|=1$
	\end{algorithmic}
\end{algorithm}

	\begin{table}[t]
	\footnotesize
		\centering
		\caption {System Parameters}
			\begin{tabular}{lll}
				\hline
				Parameters & Symbol & Value\\ \hline \hline
				Network service coverage & $\Psi$ & $1\times1$ km$^2$\\
				Density of DUs & $\Lambda_{\D}$ & $20$ DUs/km$^2$\\
				Arrival rate of URLLC, eMBB & $\lambda_{\U},\lambda_{\e}$ & $10,100$ packets/s\\
				Total processing rate of CU & $S_{\C}$ & $10^{9}$ cycles/s\\
				CPU processing of URLLC & $l_{\U}$ & $2500$ cycles/packet\\
				CPU processing for eMBB & $l_{\e}$ & $50000$ cycles/packet\\
				Pathloss exponent & $\alpha$ & $4$ \\
				DU Transmit antenna gain & $G_{c}$, $G_{c'}$ & $\{0, 20\}$ dBi\\
				Noise power & $\sigma^2$ & $-90$ dBm\\
				Delay requirement of URLLC & $t_{\U}^{(max)}$ & $1$ ms\\
				Threshold of reliability constraint & $\delta_{\U}$ & $0.99999$\\
				Successful transmission threshold& $T$ & $0$ dB\\
				Duration of TTI & $t_{\tti}(\s)$ & $0.0625$ ms\\
				Duration of retransmission & $t_{\re}$ &  $0.1$ ms\\
				\hline
			\end{tabular} \label{syspara}
	\end{table}

\section{Performance Evaluation}
\label{SIM_ANA}

We evaluate the system performance by conducting simulations, where the system parameters are listed in Table \ref{syspara}. We consider a network with the service coverage of $1 \times 1 \text{ km}^{2}$. We deploy a single CU connected to the DUs with the density of $20 \text{ DUs/km}^{2}$. The traffics of URLLC and eMBB services are defined as their respective packet arrival rates of $\lambda_{\U}=10$ and $\lambda_{\e}=100$ packets/s. The required CPU processing rate for URLLC and eMBB traffics are given by $l_{\U}=2500$ and $l_{\e}=50000$ CPU cycles/packet, respectively. The powerful CU possesses the total processing rate of $S_{\C}=10^{9}$ cycles/s. Without loss of generality, the CU with more data to be handled compared to DUs requires higher processing rate than DUs \cite{delay_constrained_fs}. Therefore, we define the DU's processing rate of $S_{\D}=\eta S_{\C}$, where $\eta \in(0,1)$ is the efficiency factor. We consider that under CM mode enabled by CoMP, the user can be served by multiple DUs, whilst the user is served by a single under a non-CoMP DM mode. According to the definition of URLLC requirement \cite{multiservice}, the reliability of URLLC is defined as the success probability of packet transmissions within $t_{\U}^{(max)}=1$ ms to be higher than $\delta_{\U}=0.99999$. The SINR threshold for successful transmission is set as $T=0$ dB. The duration of TTI is set to $t_{TTI}(\s)=0.0625$ ms as specified in \cite{num}, which is regarded as an appropriate value for the lowest slot duration especially for URLLC service. The retransmission time is set to be $t_{\re}=0.1$ ms. Note that the main burden falls in the packet processing, retransmission as well as control overhead of CoMP, which possess comparably larger latency than TTI. In the followings, we evaluate the proposed CFSMA mechanism in the proposed novel FFS framework for mixed services enabled by PS scheduling and CoMP transmissions. We will elaborate the performance validation of FFS-enabled CoMP, effect of processing capability of CU/DUs, effect of different densities of mixed serving users, and outage performance of proposed FFS-based network.

\begin{figure}[!t]
\begin{minipage}{0.49\textwidth}
	\centering
		\includegraphics[width=3in]{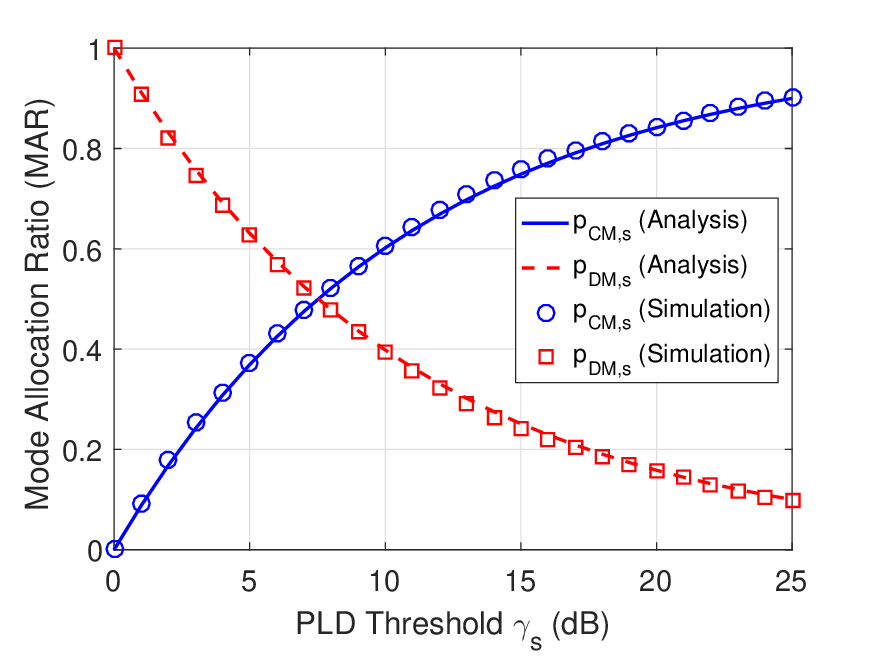}
		\caption{The results of simulation and theoretical analysis in terms of MAR considering different PLD thresholds under FFS-enabled CoMP transmission.} 		\label{fig:pld2ratio}
\end{minipage}
\quad
\begin{minipage}{0.49\textwidth}
	\centering
		\includegraphics[width=3in]{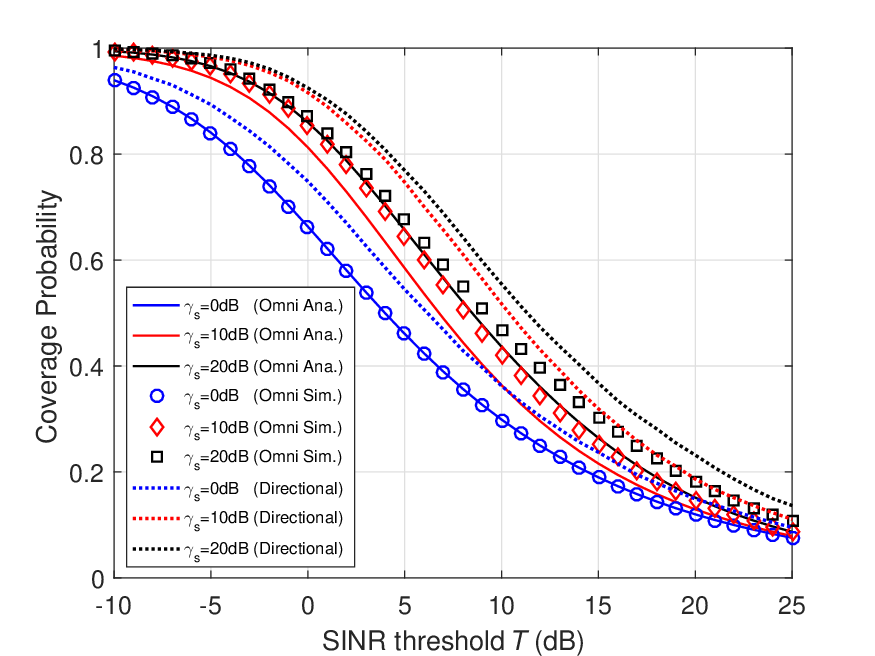}
		\caption{The performance of coverage probability of simulation and theoretical analysis considering different SINR and PLD thresholds under FFS-enabled CoMP transmission using omni/directional antennas.}
		\label{fig:coverage rate}
\end{minipage}		
\end{figure}

	\begin{figure}[!t]
		\centering
		\includegraphics[width=3in]{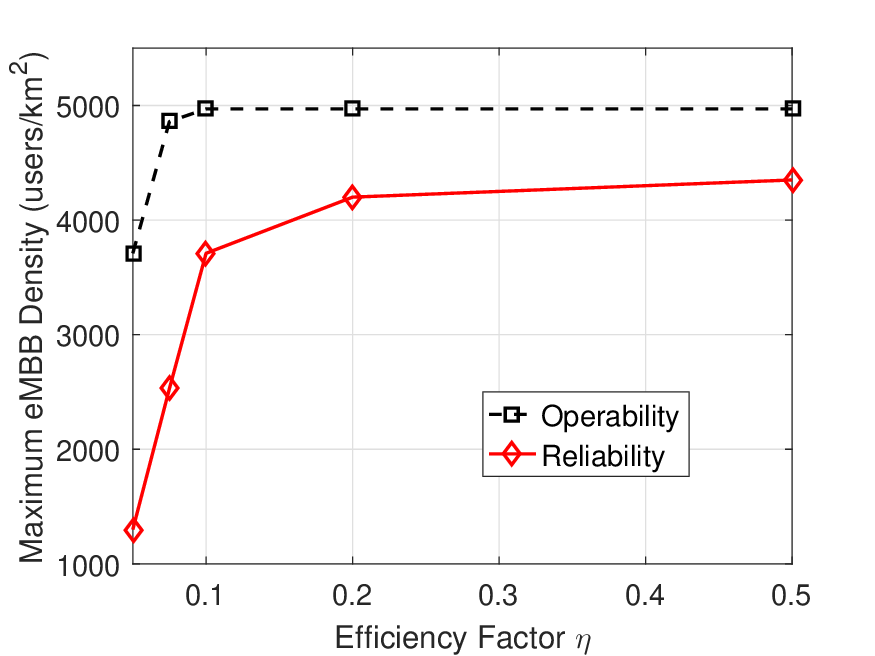}
		\caption{The maximum allowable eMBB density of mixed services restricted by operability and reliability requirements considering different efficiency factors of DUs of $\eta=\{0.05, 0.075,0.1,0.2,0.5\}$ and fixed URLLC user density of $\Lambda_{\U}=5000$ users/km$^{2}$.}
		\label{fig:maximum}
	\end{figure}

\subsection{Performance Validation of FFS-Enabled CoMP}

	In Fig. \ref{fig:pld2ratio}, we first validate the MAR results of FFS-enabled CoMP scheme considering various PLD thresholds $\gamma_{\s}\in \left[ 0 ,25\right]$ dB. Note that CoMP occurs when two nearest DUs are simultaneously serving users under CM mode, whilst non-CoMP transmission is applied for DM mode, where the probability of MAR of FS mode is derived from \eqref{eq:coverage}. We can observe that the simulation results have validated the theoretical analysis under both DM and CM modes. Furthermore, with increments of $\gamma_{\s}$, we have monotonically escalating MAR of CM $p_{\CM,\s}$ and decreasing DM of $p_{\DM,\s}$ based on property $\eqref{operable}$, i.e., $p_{\CM,\s}+p_{\DM,\s}=1$. This is because higher PLD threshold leads to a relaxation of finding two feasible nearest CoMP DUs according to $\eqref{eq:du_assignment}$, i.e., more CoMP DUs can be activated to serve users with high rate demands which potentially provokes higher MAR of CM probability.
	
	As depicted in Fig. \ref{fig:coverage rate}, FFS-enabled CoMP coverage probability is performed considering different SINR thresholds $T\in \left[-10, 25 \right]$ dB and PLD thresholds $\gamma_{\s}=\{0,10,20\}$ dB using omni-directional (unit gain) and directional antennas with $20$ dBi. Note that $\gamma_{\s}=0$ is regarded as the network enabled by non-CoMP transmission. We can observe that the simulation results of FFS-enabled CoMP asymptotically approaches our theoretical analysis derived in Propositions \ref{prop1} and \ref{prop2} and Theorem \ref{theorem:2}. It can also be seen that a lower serving coverage probability is acquired with more stringent SINR threshold. Moreover, we have higher coverage probability under higher PLD threshold $\gamma_{\s}$ due to more DUs involved in CoMP transmissions. We can have around $0.1$ to $0.15$ coverage probability enhancement by comparing directional antennas to omni-antenna. It is worth mentioning that a small difference between simulation and analysis is found under CoMP technique with $\gamma_{\s}=\{10,20\}$ dB. This is due to the reason that the covered finite interference is generated compared to the assumption of infinitely covered interference in the theory of stochastic geometry.

\subsection{Effect of Processing Capability}

	As demonstrated in Fig. \ref{fig:maximum}, we evaluate the proposed FFS-enabled CoMP network for mixed services in terms of the maximum allowable eMBB service density constrained by operability and reliability with different efficiency factors $\eta=\{0.05, 0.075,0.1,0.2,0.5\}$ and fixed URLLC user density of $\Lambda_{\U}=5000$ users/km$^{2}$. Note that system operability is defined in $\eqref{eq:st2}$, whilst the reliability constraint is according to $\eqref{eq:st1}$ under operable solution set of $\eqref{eq:st2}$. We can observe from Fig. \ref{fig:maximum} that with higher $\eta$, the system can support higher number of eMBB users, i.e., more powerful DUs having higher processing rate can alleviate the traffic loads and deal with high service demands. We can also know that the operable FFS-enabled system reaches its service limit when $\eta=0.1$ since the the sojourn time is substantially mitigated at DUs with sufficiently high processing rates. Under $\eta=0.5$, it can achieve the service number of around $5000$ eMBB users considering only operability constraint. Furthermore, with stringent service under reliability restriction, the FFS-enabled system is capable of supporting about $4400$ users. It can be induced that the processing capability have a compelling impact for reliability-aware service, i.e., low processing rate of DU potentially leads to longer E2E delay and corresponding low reliability.
	
\begin{figure*}[!t]
		\centering
		\subfigure[]
		{\includegraphics[width=2 in]{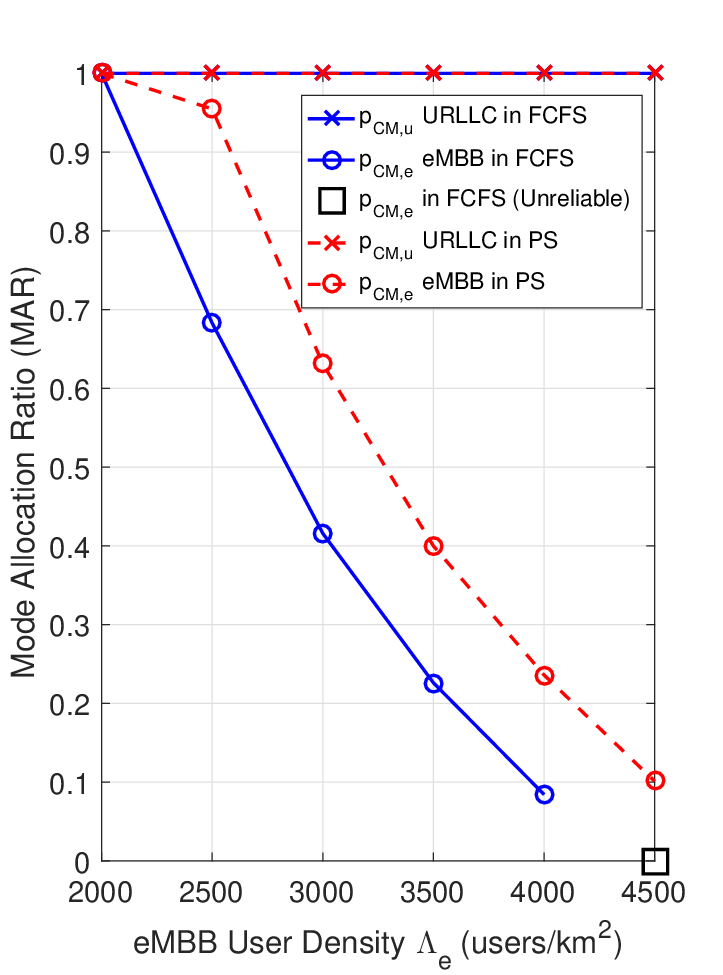} \label{fig:embb1}}
		\subfigure[]
		{\includegraphics[width=2 in]{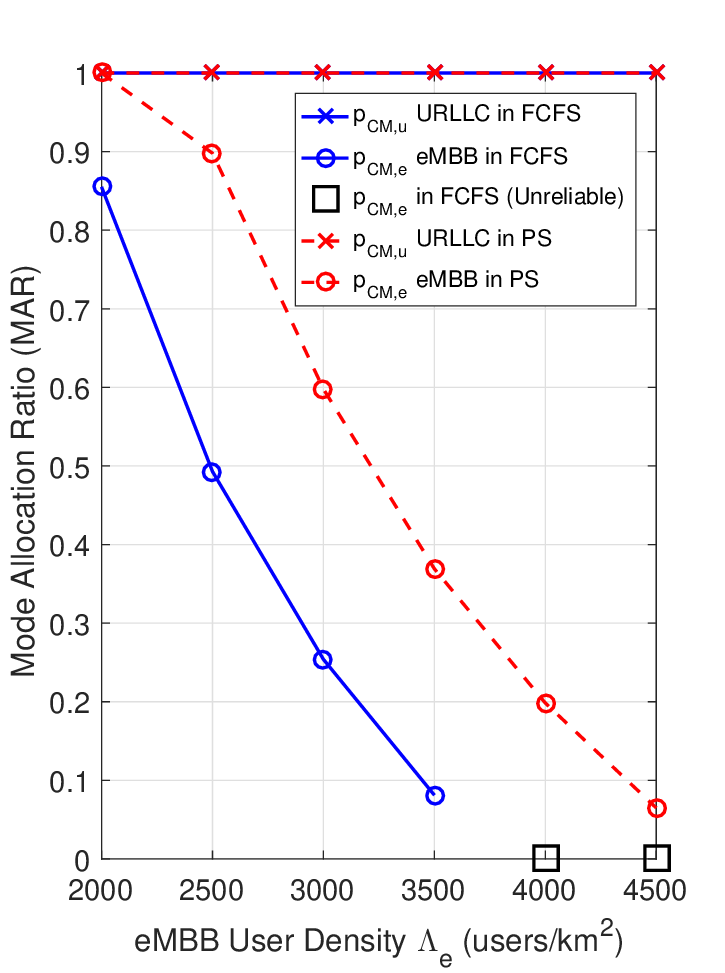} \label{fig:embb2}}
		\subfigure[]
		{\includegraphics[width=2 in]{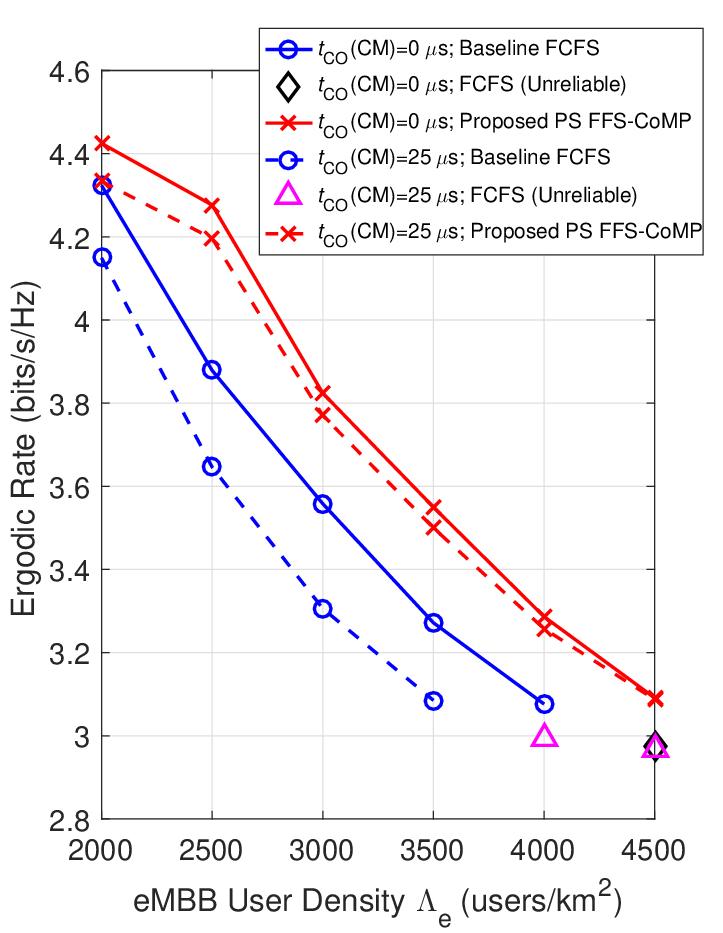} \label{fig:embb3}}
		\caption{The performance results of MAR and ergodic rates in an FFS-enabled CoMP network employing PS mechanism compared to the FCFS baseline under DU processing efficiency factor $\eta=0.2$ and URLLC user density of $\Lambda_{\U}=5000$ users/km$^{2}$. (a) MAR without CM overhead $t_{\co}(\CM)=0\ \mu$s (b) MAR with CM overhead $t_{\co}(\CM)=25 \ \mu$s (c) Ergodic rate.}
		\label{fig:embb_total}
	\end{figure*}	 

\begin{figure*}[t]
		\centering
		\subfigure[]
		{\includegraphics[width=2 in]{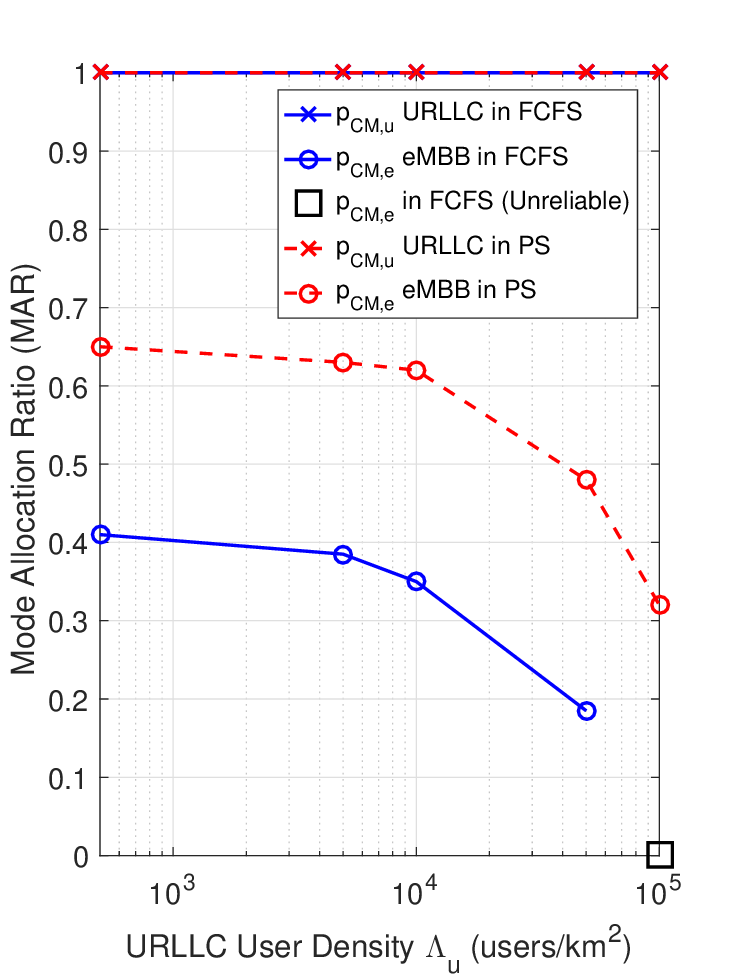} \label{fig:urllc1}}
		\subfigure[]
		{\includegraphics[width=2 in]{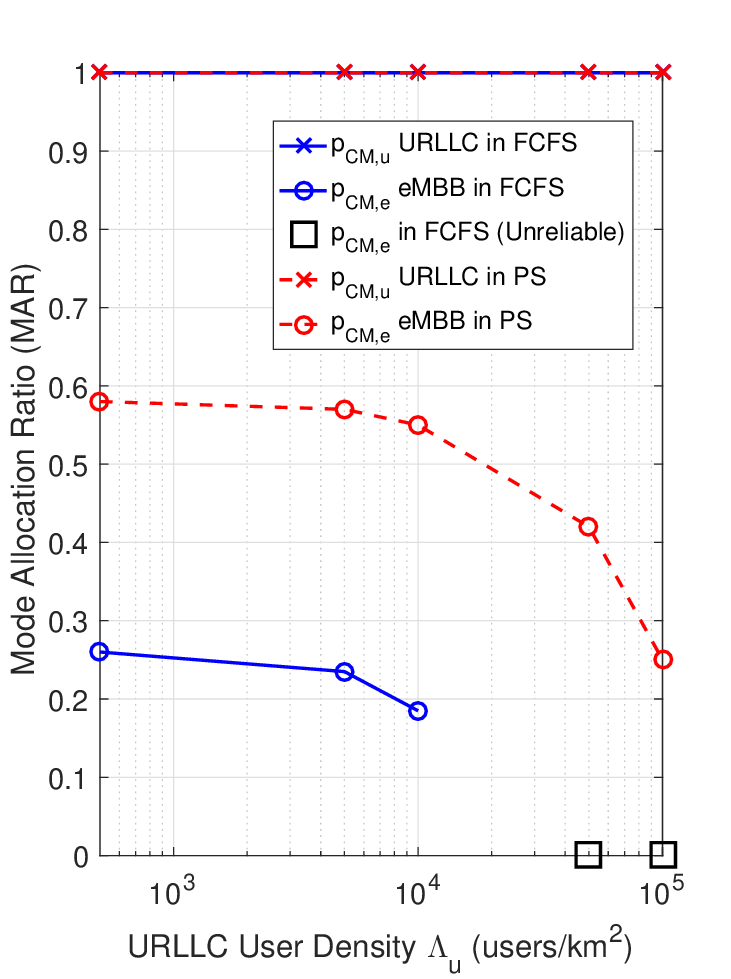} \label{fig:urllc2}}
		\subfigure[]
		{\includegraphics[width=2 in]{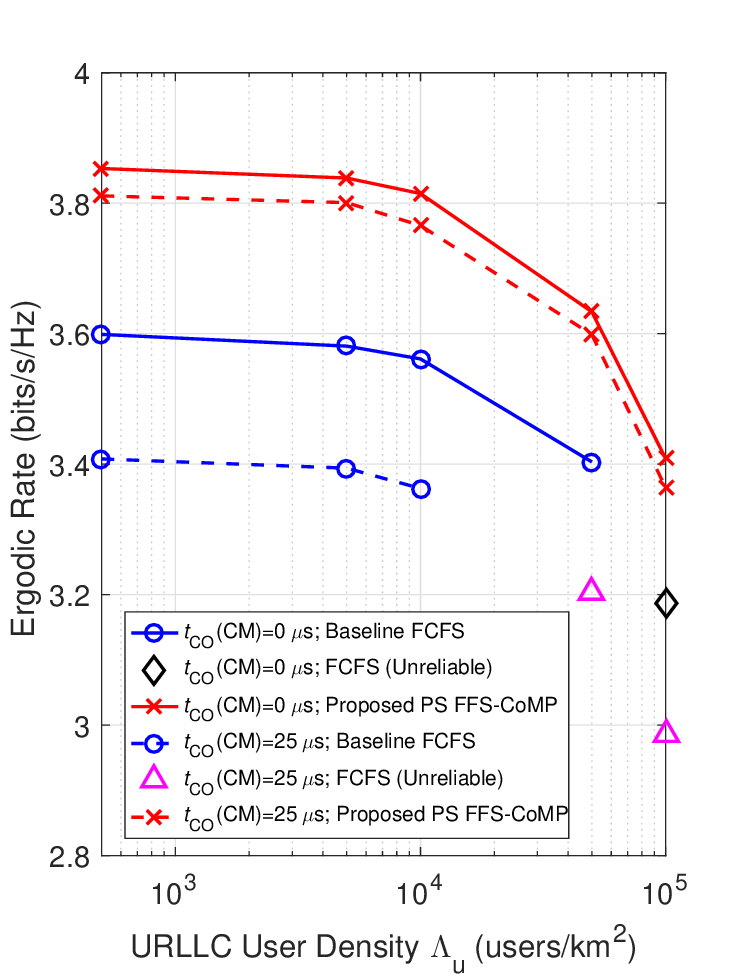} \label{fig:urllc3}}
		\caption{The performance results of MAR and ergodic rates in an FFS-enabled CoMP network employing PS mechanism compared to the FCFS baseline under DU processing efficiency factor $\eta=0.2$ and eMBB user density of $\Lambda_{\e}=3000$ users/km$^{2}$. (a) MAR without CM overhead $t_{\co}(\CM)=0\ \mu$s (b) MAR with CM overhead $t_{\co}(\CM)=25\ \mu$s (c) Ergodic rate.}
		\label{fig:urllc_total}
	\end{figure*}

\subsection{Effect of User Densities Under Mixed Services}

	As depicted in Figs. \ref{fig:embb_total} and \ref{fig:urllc_total}, we evaluate the performance of MAR ratio and ergodic rate in an FFS-enabled CoMP networking employing PS mechanism under DU processing efficiency factor of $\eta=0.2$ with different eMBB and URLLC user densities, respectively. We conduct simulations in terms of varying control overhead of CM mode and compare our proposed PS-based FFS framework to the baseline of conventional FCFS method. As demonstrated in Fig. \ref{fig:embb_total}, we consider fixed URLLC user density as $\Lambda_{\U}=5000$ users/km$^{2}$. We can observe from Figs. \ref{fig:embb1} and \ref{fig:embb2} that the MARs of $p_{\CM,\U}$ of both FCFS and PS mechanisms for URLLC services are identically flat and at their maximal values. This is due to the reason that the proposed FFS-enabled system should be operated at full CM mode for guaranteeing all prioritized URLLC services by satisfying the reliability constraint and offloading some eMBB services to DM mode. Accordingly, as seen from both figures, the MAR of $p_{\CM,\e}$ for eMBB users is monotonically decreasing with the increment of eMBB user density. In other words, it results in a full DM mode using conventional non-CoMP transmission for eMBB users under an overloaded network reaching its limits of processing and communication resources. By comparing the scenarios without and with signalling overhead respectively in Figs. \ref{fig:embb1} and \ref{fig:embb2}, it can be inferred that the PS-FFS-enabled network is capable of providing more CM-based eMBB services under low signalling overhead due to its low E2E delay derived in $\eqref{delay}$, e.g., around $1920$ eMBB users in $\Lambda_{\e}=3000$ users/km$^{2}$ are served under PS-enabled CM-based CoMP transmission which provides $120$ more users with high-quality services compared to that with high CM overhead. Moreover, as a benefit of simultaneous processing capability, the PS-enabled network possesses higher MAR values to support more CM-based services compared to that using conventional FCFS method, i.e., the FCFS-based network becomes unreliable when there exist more than $4500$ and $4000$ eMBB users without and with control overhead of CM as respectively shown in Figs. \ref{fig:embb1} and \ref{fig:embb2}. This is because that FCFS processing is mostly occupied by prioritized URLLC services with small amount of available processing resources assigned to eMBB users, which induces an unreliable system under overloaded eMBB services. Moreover, as observed from Fig. \ref{fig:embb3}, the ergodic rate performance becomes decreased with the escalating density of eMBB users due to low MAR $p_{\CM,\e}$ under compellingly increased traffic loads of an FFS network. The ergodic rate saturates when $\Lambda_{\e}\geq 4000$ users/km$^{2}$ since the system reaches its limits to provide resources for eMBB users. The higher CM overhead will lead to lower rate performance since it has higher E2E delay which requires more powerful processing capability to maximize eMBB rate while sustaining the reliability of URLLC services.

	As shown in Fig. \ref{fig:urllc_total}, we evaluate MAR and ergodic rate of different URLLC user densities of $\Lambda_{\U}=\{5\times10^{3}, 10^{4}, 5\times10^4, 10^4, 5\times10^4, 10^5\}$ users/km$^{2}$ considering fixed eMBB user density as $\Lambda_{\e}=3000$ users/km$^{2}$. As observed from Figs. \ref{fig:urllc1} and \ref{fig:urllc2}, the MAR values of URLLC in both FCFS and PS mechanisms are $p_{\CM,\U}=1$ since the reliability of prioritized URLLC services is guaranteed by compromising resources from eMBB users. Therefore, it has demonstrated the similar results as that in Figs. \ref{fig:embb1} and \ref{fig:embb1} that MAR of eMBB $p_{\CM,\e}$ decreases with compelling traffic loads of either eMBB or URLLC service. Beneficial to FFS-enabled network, the PS mechanism is capable of providing higher MAR to eMBB users, i.e., CoMP transmission is activated more frequently with better signal quality compared to that of FCFS method. Moreover, FCFS is unable to offer resilient services, which leads to an unreliable system when $\Lambda_{\U}=10^{5}$ users/km$^{2}$ without CM overhead and when $\Lambda_{\U}= 5\times 10^{4}$ users/km$^{2}$ with CM overhead. Furthermore, as depicted in Fig. \ref{fig:urllc3}, we can observe that the ergodic rate performance deteriorates with increasing traffic loads of URLLC services and higher CM overhead due to limited processing and communication resources. Additionally, under $t_{\co}(\CM)=25 \ \mu$s, the unreliable system utilizing FCFS scheduling has substantially low rate of around $3$ bits/s/Hz compared to that of around $3.4$ bits/s/Hz in proposed PS-enabled FFS-CoMP network. To summarize, the proposed PS-enabled FFS-CoMP network achieves the highest rate performance due to outstanding scheduling capability and resilient processing and traffic assignment benefited from flexible functional split operation.

\begin{figure*}[!t]
		\centering
		\subfigure[]
		{\includegraphics[width=2 in]{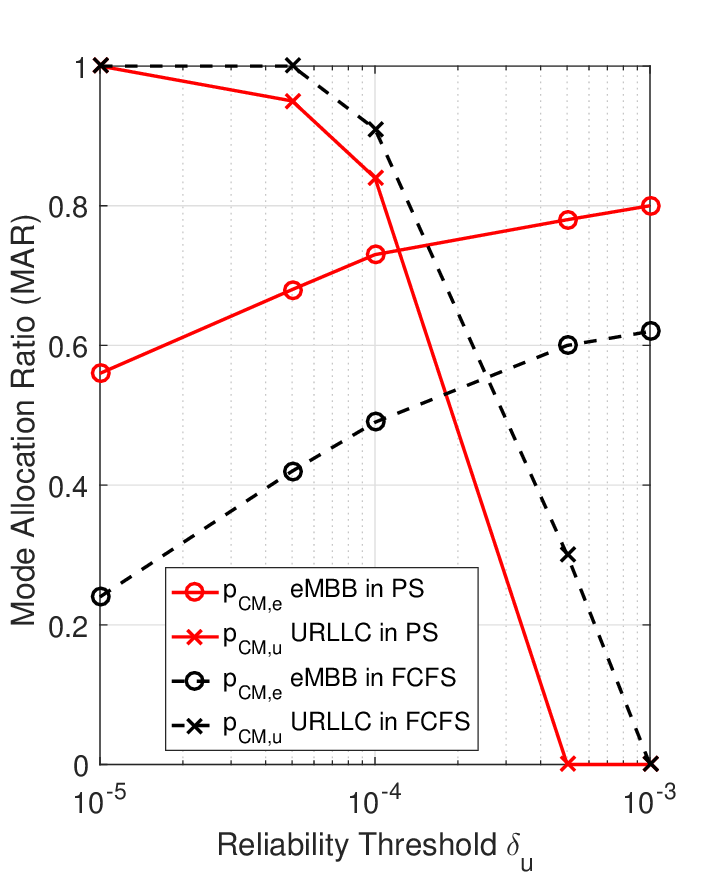} \label{fig:relia1}}
		\subfigure[]
		{\includegraphics[width=2 in]{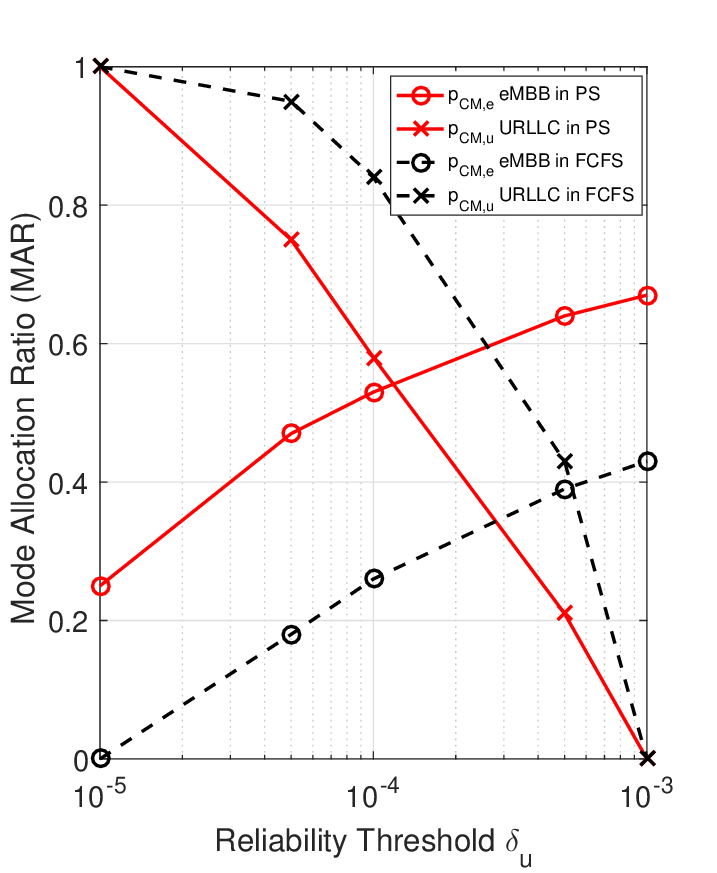} \label{fig:relia2}}
		\subfigure[]
		{\includegraphics[width=2 in]{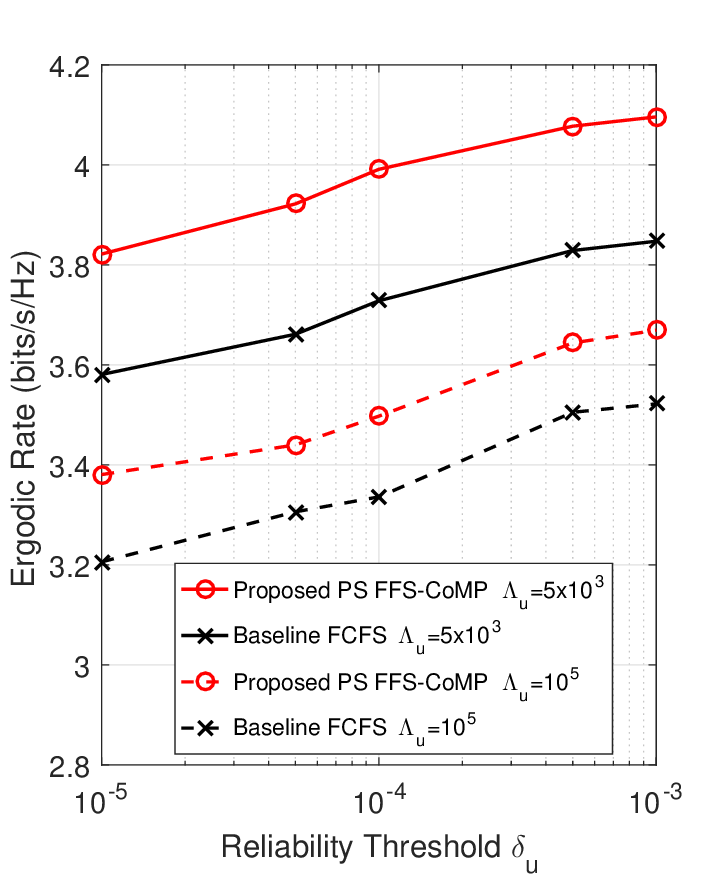} \label{fig:relia3}}
		\caption{The performance results of different reliability requirements of $\delta_{\U}=\{10^{-5}, 5\times 10^{-5}, 10^{-4}, 5\times10^{-4}, 10^{-3}\}$ under a PS-enabled FFS-CoMP network compared to FCFS considering DU processing efficiency factor $\eta=0.2$, eMBB user density of $\Lambda_{\e}=3000$ users/km$^{2}$, and CM overhead $t_{\co}(\CM)=25 \mu$s. (a) MAR with $\Lambda_{\U}=5\times10^3$ URLLC users/km$^{2}$ (b) MAR with $\Lambda_{\U}=10^5$ URLLC users/km$^{2}$ (c) Ergodic rate.}
		\label{fig:relia_total}
	\end{figure*}

\subsection{Reliability Thresholds}

In Fig. \ref{fig:relia_total}, we simulate the performance considering different reliability requirements of $\delta_{\U}=\{10^{-5}, 5\times 10^{-5}, 10^{-4}, 5\times10^{-4}, 10^{-3}\}$ for a PS-enabled FFS-CoMP network with efficiency factor $\eta=0.2$, $\Lambda_{\e}=3000$ eMBB users/km$^{2}$ and CM overhead $t_{\co}(\CM)=25 \ \mu$s. As observed from Figs. \ref{fig:relia1} and \ref{fig:relia2}, it can be deduced that smaller MAR values of CM are acquired for URLLC services due to relaxed reliability constraints, i.e., full DM mode is sufficient to guarantee reliability of URLLC users with $\delta_{\U}=10^{-3}$. Accordingly, with more processing and communication resource remained, more eMBB users can be served under CM mode using CoMP transmission. Moreover, benefited by proposed PS scheduling, URLLC services with small-sized packets can be instantly processed, which provides higher MAR values for eMBB users than that of inflexible FCFS method. We can also see that fewer eMBB users are served in CM mode with lower MAR ratios when a large number of URLLC users are queued. For example, by observing the PS method from Figs. \ref{fig:relia1} and \ref{fig:relia2}, a degradation from $p_{\CM,\e}=0.8$ to $p_{\CM,\e}=0.67$ takes places under $\delta_{\U}=10^{-3}$ for the increased number of URLLC users from $\Lambda_{\U}=5\times10^3$ to $\Lambda_{\U}=10^5$ users/km$^{2}$. Therefore, as demonstrated in Fig. \ref{fig:relia3}, more URLLC users potentially lead to lower ergodic rate performance due to congested traffic and limited processing capability. However, with less stringent service restrictions, higher rate can be achieved by providing higher MAR in CM mode for eMBB users while guaranteeing URLLC reliability. Moreover, owing to its advantages of flexible functional split mode allocation, the proposed PS-based FFS-CoMP network achieves higher ergodic rate compared to the baseline of FCFS.

\subsection{Outage Performance of FFS-based Network}

	Fig. \ref{fig:CO} depicts the impact of CM control overhead on both E2E delay outage and URLLC MAR $p_{\CM,\U}$ in the proposed PS-enabled FFS-CoMP network, where relevant parameters are set as $\eta=0.2$, $\Lambda_{\U}=5000$ and $\Lambda_{\e}=3000$ users/km$^{2}$. We consider the case of fixed CM overhead in Fig. \ref{fig:CO1} and varying CM overhead in Fig. \ref{fig:CO2}, where we define $\epsilon=\{0,0.25,0.5,0.75,1\}$ as the ratio of two CM overhead of $t_{\co}(\CM)=\{25,100\}$ $\mu$s, i.e., $\epsilon=1$ means that the FFS system has the only high overhead of $100$ $\mu$s. 
For the case of fixed CM overhead in Fig. \ref{fig:CO1}, we can deduce that higher CM overhead will lead to higher delay outage due to longer CM-based CoMP processing time. However, it can be seen that URLLC MAR of CM is one when $t_{\co}(\CM)\leq 75$ $\mu$s, whilst zero MAR is derived when $t_{\co}(\CM)=\{100,125\}$ $\mu$s. This is due to the reason that under low-overhead the system tends to quickly process the prioritized URLLC services. Nevertheless, for the high-overhead scenario, arbitrary assignment of any CM mode will guarantee no reliability for URLLC services, and therefore it will allocate those traffic in DM mode resulting in a high delay outage. In Fig. \ref{fig:CO2}, we can observe that the case of varying control overhead reveals similar monotonic increasing outage as that in the fixed case. The small portion of high-overhead will substantially deteriorates the outage performance as found from $\epsilon=0$ to $\epsilon=0.25$. Moreover, it becomes unable to meet the URLLC requirement of $\delta_{\U}=1-10^{-5}$ when $\epsilon\geq 0.25$ due to compellingly high CM overhead. Due to the flexible assignment of an FFS-enabled network, the system is capable of offering better URLLC services with high MAR values under two types of CM overhead. However, with larger $\epsilon$, lower MAR is acquired to support near DM service with limited processing and communication resources.

	\begin{figure*}[!t]
		\centering
		\subfigure[]{
		\includegraphics[width=3in]{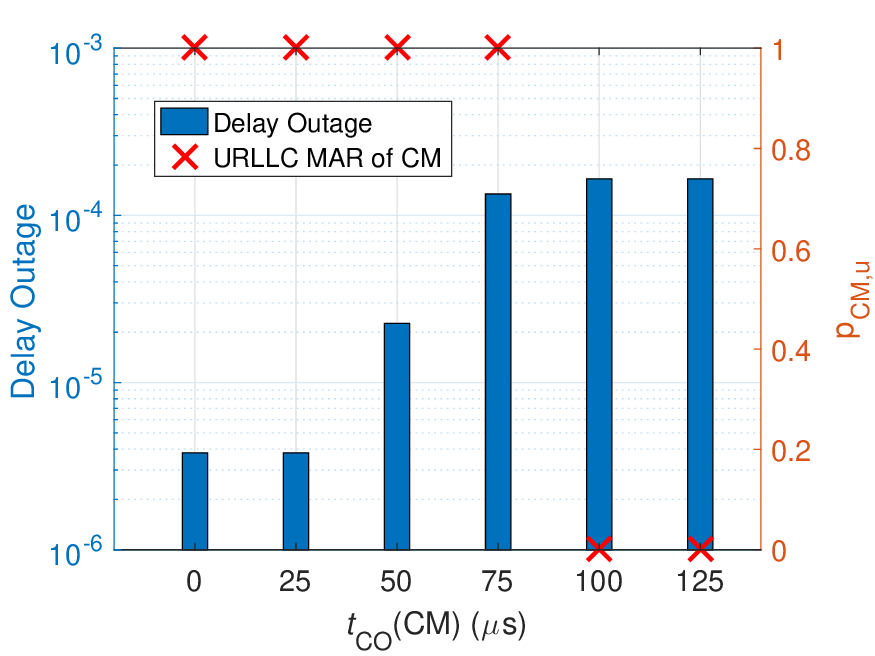} \label{fig:CO1}}
		\subfigure[]{
		\includegraphics[width=3in]{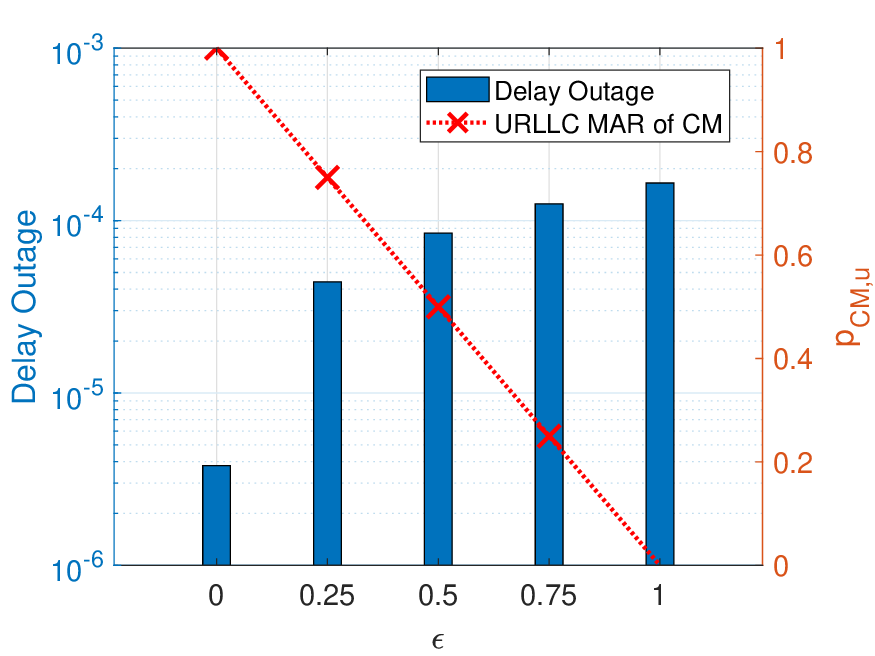} \label{fig:CO2}}
		\caption{The result of delay outage and URLLC MAR of CM in proposed PS-FFS-CoMP network considering varying CM control overhead with DU processing efficiency factor $\eta=0.2$, $\Lambda_{\U}=5000$ and $\Lambda_{\e}=3000$ users/km$^{2}$. (a) Fixed control overhead (b) Varying control overhead with $\epsilon=\{0,0.25,0.5,0.75,1\}$ for the ratio of $t_{\co}(\CM)=\{25,100\}$ $\mu$s.}
		\label{fig:CO}
	\end{figure*}	

	\begin{figure*}[!t]
		\centering
		\subfigure[]
		{\includegraphics[width=3in]{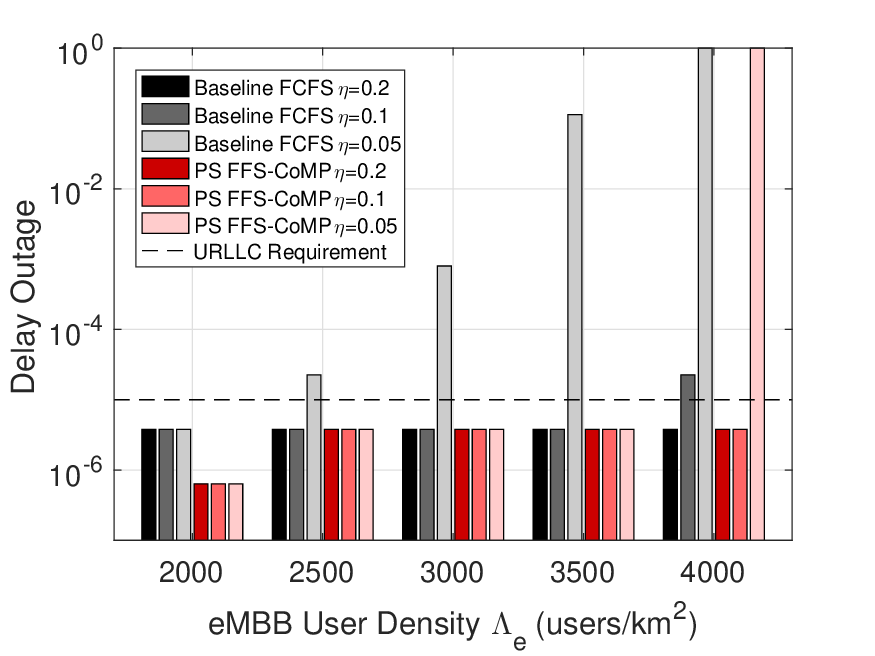}\label{fig:eff1}}
		\subfigure[]
		{\includegraphics[width=3in]{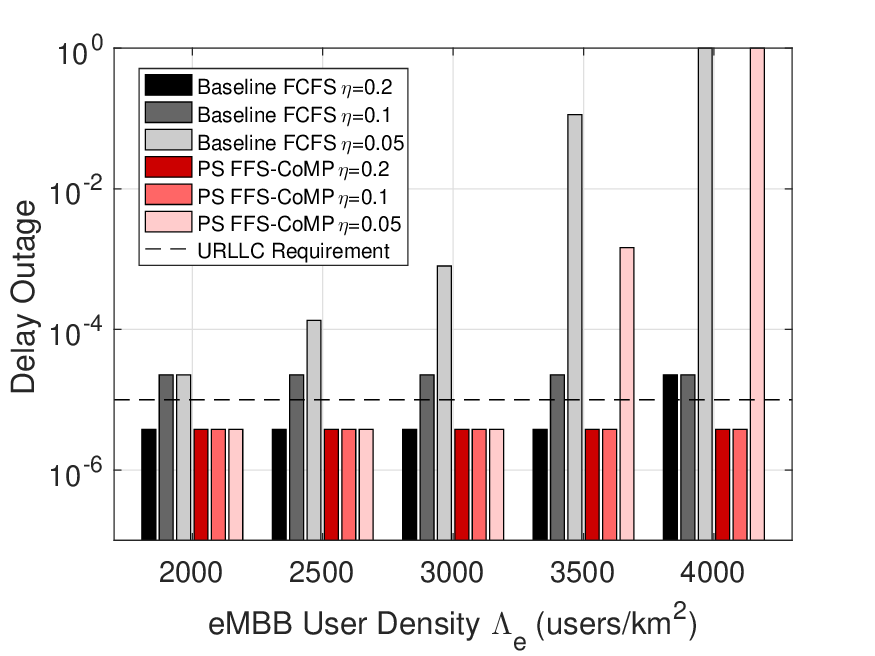} \label{fig:eff2}}
		\caption{The delay outage performance of proposed PS FFS-CoMP network in terms of different eMBB user densities and DU efficiency factors of $\eta=\{0.05,0.1,0.2\}$ with URLLC user density of $\Lambda_{\U}=5000$ users/km$^{2}$. (a) CM overhead $t_{\co}(\CM)=0$ $\mu$s (b) CM overhead $t_{\co}(\CM)=25$ $\mu$s.} \label{fig:eff}
	\end{figure*}

As shown in Fig. \ref{fig:eff}, we evaluate the delay outage performance of proposed PS FFS-CoMP network in terms of different eMBB user densities and DU efficiency factors of $\eta=\{0.05,0.1,0.2\}$ with URLLC user density of $\Lambda_{\U}=5000$ users/km$^{2}$ considering CM overhead of $t_{\co}(\CM)=\{0, 25\}$ $\mu$s in Figs. \ref{fig:eff1} and \ref{fig:eff2}, respectively. We can observe that the delay outage increases with the increments of eMBB traffic and lower DU processing capability. This is because when $\eta$ is compellingly smaller, the packet sojourn time in low-capability DU will become longer which deteriorates the delay performance. 
It is also known that the proposed framework can mostly satisfy the specified URLLC requirements of delay outage $10^{-5}$ with the low-overhead case, which is demonstrated in Fig. \ref{fig:CO1}. Nevertheless, for $t_{\co}(\CM)=25$ $\mu$s in Fig. \ref{fig:CO2}, the reliability constraint cannot be fulfilled once the DU processing capability is as low as $\eta=0.05$. However, benefited by simultaneous processing employing PS scheduling, the proposed PS FFS-CoMP network is able to instantly process the small-sized URLLC packets, which outperforms the inflexible FCFS method in terms of much lower delay outage.

	In Fig. \ref{fig:inst}, it reveals the delay outage performance difference between averaged optimized and instantaneous mechanisms. Note that the averaged scheme is derived and conducted as proposed in this paper, whilst the instantaneous policy considers time-slot based traffic and processing capability with both non-prioritized and priority-based scheduling. In non-prioritized use case, all URLLC users are equivalently assigned PS-enabled processing capability. On the contrary, in priority-based scheduling, a greedy-based method is adopted to firstly assign users with top-priority having more restricted latency requirements. As observed from Fig. \ref{fig:inst}, due to incapability of assigning processing, there is no performance difference for FCFS under either averaged or instant evaluation. Moreover, owing to quantization error from integer values of analysis of failure re-transmissions, instantaneous policy of non-prioritized scheduling has a slightly lower delay outage, which potentially reflects more realistic traffics and scenarios. Benefited by priority-based scheduling, it is capable of providing resilient PS allocation for eMBB users as well as URLLC users with urgent needs, resulting in the lowest delay outage.
	
\begin{figure}[!t]
		\centering
\includegraphics[width=3in]{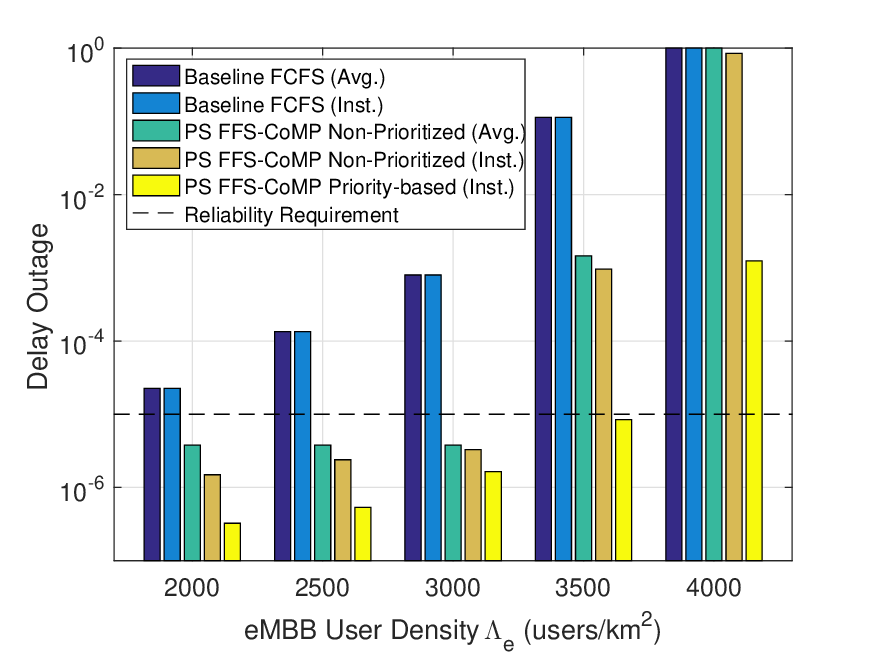}
		\caption{The delay outage of proposed PS FFS-CoMP network and baseline FCFS under averaged optimized (Avg.) and instantaneous (Inst.) cases considering $\eta=0.05$, $t_{\co}(\CM)=25$ $\mu$s, and $\Lambda_{\U}=5000$ users/km$^{2}$.} \label{fig:inst}
	\end{figure}
	
\subsection{Benchmark Comparison}

	\begin{figure*}[!t]
		\centering
		\subfigure[]
		{\includegraphics[width=3in]{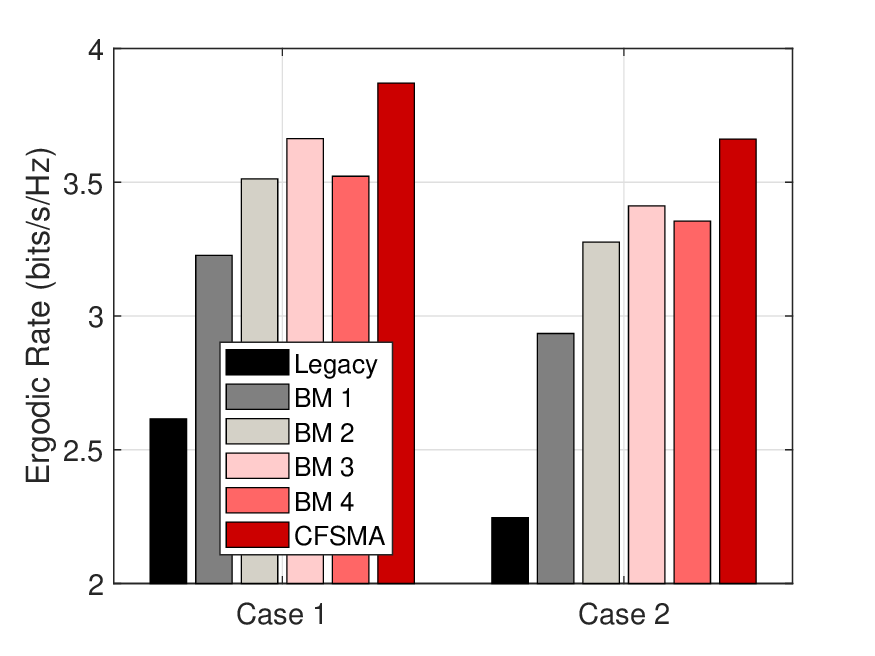}\label{ben1}}
		\subfigure[]
		{\includegraphics[width=3in]{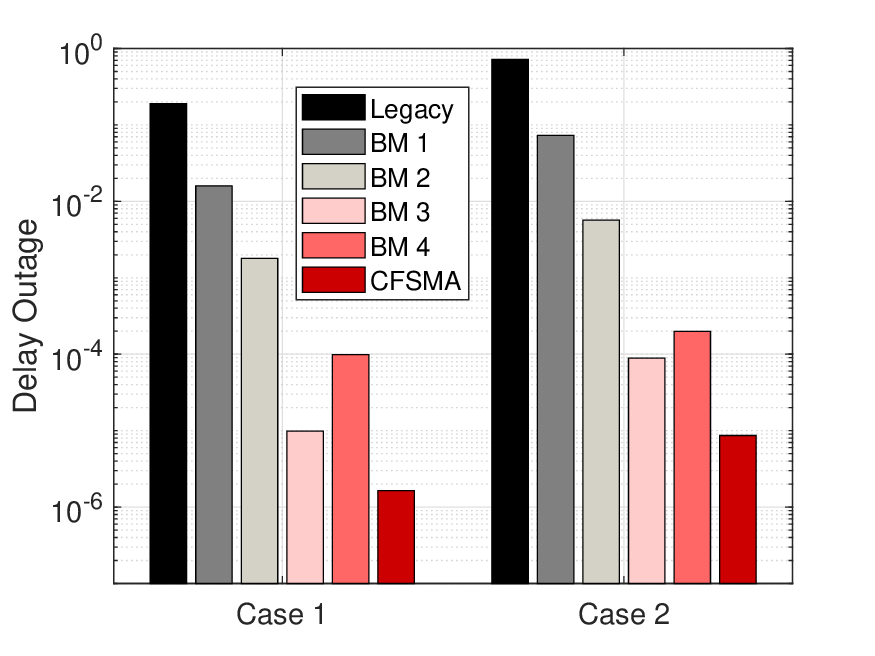} \label{ben2}}
		\caption{The performance comparison of proposed CFSMA scheme in PS FFS-CoMP and benchmarks of legacy network and different fixed functional split options in terms of eMBB user densities $\Lambda_{\e}=3000$ users/km$^{2}$ and DU efficiency factors of $\eta=0.2$ and CM overhead $t_{\co}(\CM)=25$ $\mu$s. Case 1 stands for URLLC user density of $\Lambda_{\U}=5\times10^3$ URLLC users/km$^{2}$, whereas Case 2 is for $\Lambda_{\U}=10^5$ URLLC users/km$^{2}$. (a) Ergodic rate (b) Delay outage.} \label{ben}
	\end{figure*}

In Fig. \ref{ben}, we have conducted the rate and delay outage performances of proposed CFSMA scheme in PS FFS-CoMP, which is compared with the benchmarks (BMs) in open literature of legacy network and different fixed functional split options in terms of eMBB user densities $\Lambda_{\e}=3000$ users/km$^{2}$ and DU efficiency factors of $\eta=0.2$ and CM overhead $t_{\co}(\CM)=25$ $\mu$s. We consider two cases: Case 1 stands for sparse URLLC users with $\Lambda_{\U}=5\times10^3$ URLLC users/km$^{2}$, whereas Case 2 is for dense URLLC users with $\Lambda_{\U}=10^5$ URLLC users/km$^{2}$. The description of benchmarks is listed as follows: 
\begin{itemize}
\item \textbf{Legacy} network in \cite{r1, r2} indicates a conventional architecture with a BS containing all functions without any split options with all mixed traffics processed at BS under FCFS mechanism.
\item \textbf{BM 1} is deployed with fixed FSO 4 \cite{com} in only DM mode, which CoMP is inactivated due to high-MAC layer at DUs. Note that FCFS scheduling is applied.

\item \textbf{BM 2} in \cite{com2} is performed under fixed FSO 7.2 in pure CM mode with all traffics under CoMP-enabled network. Note that FCFS scheduling is utilized.

\item \textbf{BM 3} considers PS \cite{PS} scheduling under FSO 7.2.

\item \textbf{BM 4} in \cite{capacity_constrained_fs} has also designed a flexible FSO method but with limited function options. Note that FCFS is adopted in the paper.
\end{itemize}
We can observe from Fig. \ref{ben} that under mixed services the legacy network performs the worst rate and highest delay outage due to all traffics congested at BS. For BM 1, improved rate and delay are obtained due to FSO compared to the legacy with network functions flexibly separated to CU/DUs. However, as stated in \cite{tutorial}, few benefits of FSO 4 in BM 1 are demonstrated owing to closely relation between RLC and MAC functions. Therefore, BM 2 with option 7.2, deemed to be an appropriate candidate in most of recent works \cite{tutorial}, maintains high-PHY in CU with fewer functions operated at DUs. With CoMP activated under BM 2, ergodic rate can be further improved from $3.22$ to $3.51$ bits/s/Hz. Nonetheless, conventional FCFS method congests the latency-aware traffics inducing little improvement of delay outage, which is substantially enhanced by employing PS in BM 3 with a reduction from around $1.8\times 10^{-3}$ to $1\times 10^{-5}$ in sparse case. In BM 4, flexible FSOs are proposed with improved delay outage compared to the fixed option in BM 2. Note that BMs 2 and 4 are operated under FCFS method. On the contrary, BM 3 possesses a better performance than FFS of BM 4 thanks to the PS capable of jointly processing latency-/throughput-aware packets. Benefited by both FFS and PS, CFSMA capable of selecting appropriate functions and scheduling portions for CU/DUs accomplishes the highest ergodic rate and lowest delay outage compared to the other existing benchmarks in open literature. To elaborate a little further, considering fairness metric, the worst case happens in FCFS, whilst PS-based mechanism possesses higher fairness than FCFS since almost all services are available by sharing queues. However, fixed FSO performs moderate fairness due to either unsatisfaction of eMBB traffic in FSO 4 or delay outage of URLLC in FSO 7.2.

\section{Conclusion} \label{SEC_CON}
	In this paper, we conceive a novel FFS-enabled framework for mixed eMBB-URLLC services including PS-based scheduling for hybrid CM-based CoMP and DM-based non-CoMP transmissions. The tractable analysis of PS-FFS-CoMP for mixed services in terms of ergodic rate and delay performances is provided by the employment of stochastic geometry and queuing property. Aiming at maximizing eMBB ergodic rate while guaranteeing URLLC reliability and system operability, we have proposed a CFSMA scheme to flexibly assign user-centric FSOs of CM/DM modes for multi-services. Performance validation shows that our theoretical analysis of proposed PS-FFS-CoMP framework asymptotically approaches the simulation results. With higher CM overhead, it potentially provokes lower rate and higher outage performance due to higher E2E delay. Moreover, under compellingly high traffic loads, the system utilizing conventional FCFS method cannot fulfill reliability requirement of URLLC users. Benefited from simultaneous processing for multi-service traffic, the proposed PS-enabled FFS-CoMP network is capable of providing resilient and efficient service for URLLC users while providing more computing and communication resources for improving eMBB users, which outperforms the benchmarks of conventional FCFS scheduling, non-FSO network, fixed FSOs, and limited available FSO selections in open literature.

\bibliographystyle{IEEEtran}

\footnotesize
\bibliography{IEEEabrv}

\end{document}